\documentclass[11pt]{article}

\usepackage[bibliosources=refs.bib]{pomegranate}
\usepackage{tikz}
\usetikzlibrary{arrows.meta,calc,decorations.pathreplacing,positioning}

\DeclareOperator{\pf}
\DeclareOperator{\Ent}
\DeclareOperator{\Var}
\DeclareOperator{\supp}
\DeclareOperator{\polylog}
\DeclareOperator{\diag}
\DeclareOperator{\Deg}
\DeclareOperator{\Star}
\newcommand{\eps}{\varepsilon}
\newcommand{\TV}{\mathrm{TV}}
\newcommand{\ET}{\mathrm{ET}}
\newcommand{\cN}{\mathcal N}
\newcommand{\cR}{\mathcal R}
\newcommand{\cE}{\mathcal E}
\newcommand{\cH}{\mathcal H}
\newcommand{\bits}{\{0,1\}}
\newcommand{\symdiff}{\triangle}
\NewDocumentCommand{\ip}{s m m}{%
  \IfBooleanTF{#1}{\dotprod*{#2,#3}}{\dotprod{#2,#3}}%
}
\crefname{conjecture}{conjecture}{conjectures}
\Crefname{conjecture}{Conjecture}{Conjectures}

\title{Sampling Directed Eulerian Tours in \(\widetilde O(m^{3/2})\) Time}
\author[1]{Nima Anari}
\affil[1]{Stanford University, \texttt{anari@stanford.edu}}
\date{}

\begin{document}

\maketitle

\begin{abstract}
We give a randomized algorithm that samples a nearly uniform Eulerian tour of a
directed Eulerian multigraph with \(m\) arcs in
\(\widetilde O(m^{3/2})\) time.  The guarantee is worst-case, applies to
arbitrary directed Eulerian multigraphs, and breaks the \(mn\)-type
arborescence-sampling barrier on sparse graphs.

The core case is a \(2\)-in/\(2\)-out graph.  We introduce a new local Markov
chain, the flip--repair walk: one step locally splits a tour into two circuits
and then chooses uniformly among the local flips that repair the state to one
tour.  We prove that this walk mixes in nearly linear many steps and implement
the walk using a dynamic chord data structure.  A pointwise degree-reduction
wrapper extends the sampler from this degree-two core to arbitrary degrees
while preserving the \(\widetilde O(m^{3/2})\) total running time.

The high-level algorithmic plan, the switching-network reduction, and the
dynamic data-structure argument were devised by the author.  The author
conjectured the mixing theorem underlying the analysis, and GPT 5.5 Pro
Extended produced its linear-algebra proof.  Codex assisted with manuscript
assembly and typesetting.
\end{abstract}

\section{Introduction}

Let \(G=(V,E)\) be a strongly connected directed multigraph with
\[
  \deg^+(v)=\deg^-(v)
  \qquad\text{for every }v\in V.
\]
Arcs are distinguishable, and loops are allowed; a loop contributes one incoming
and one outgoing incidence at its endpoint.  An Eulerian tour of \(G\) is a
cyclic ordering of the arcs in which consecutive arcs are head-to-tail incident
and every arc appears exactly once.  Eulerian tours are among the most classical
objects in graph theory.  In directed graphs, their enumeration is governed by
the BEST theorem of de Bruijn, van Aardenne-Ehrenfest, Smith, and Tutte
\cite{SmithTutte1941,deBruijnAardenneEhrenfest1951,Tutte1984}.  The theorem
expresses the number of tours as an arborescence count times a product of local
factorials, and it suggests a natural exact sampler: sample the arborescence
part and then sample the local orders.  This approach was used for sequence
shuffling by Kandel, Matias, Unger, and Winkler
\cite{KandelMatiasUngerWinkler1996} and sharpened using Wilson's
cycle-popping algorithm \cite{ProppWilson1998,JiangEtAl2008}.  Its running
time, however, is tied to the cost of sampling random arborescences.  On
Eulerian digraphs, the best general random-walk mixing bounds for this task
are on the \(mn\) scale \cite{BoczkowskiPeresSousi2018}, which is not nearly
linear on sparse graphs.

The goal of this paper is to break this running-time barrier, at least for
sparse graphs.  We do this by sampling tours directly rather than sampling an
arborescence and then completing it to a tour.  The state of the chain is always
a tour, and each transition changes only a small number of local successor
choices.

A convenient language for local successor choices is that of transition
systems.  A transition system chooses, at every vertex \(v\), a bijection from
the incoming arcs of \(v\) to the outgoing arcs of \(v\).  Such a choice induces
a permutation of the arc set \(E\): after an arc enters a vertex, the bijection
specifies the next arc.  The transition system is an Eulerian tour exactly when
this permutation has one cycle.  We write \(\ET(G)\) for the set of transition
systems whose induced permutation has one cycle; this is the tour convention
used throughout the paper.  An initial tour can be found in linear time by the
classical algorithm of Hierholzer \cite{Hierholzer1873}.

The basic obstruction is also the key move.  Changing one local successor
choice usually breaks a tour into two circuits.  The chain therefore makes one
such break, uses the resulting two-circuit state as an intermediate layer, and
then chooses uniformly among the local changes that glue the circuits back into
one tour.  We call this local transition the \emph{flip--repair} move.

Our main result is the following.

\begin{theorem}[Main theorem]\label{thm:main}
Let \(G\) be a strongly connected directed Eulerian multigraph with \(m\) arcs.
For every \(0<\eps<1\), there is a randomized algorithm which outputs a tour
whose distribution is within \(\eps\) total variation distance of the uniform
distribution on \(\ET(G)\).  Its running time is
\[
  O\!\left(
    m^{3/2}\polylog(m/\eps)
  \right).
\]
The algorithm uses \(O(m\polylog(m/\eps))\) space.
\end{theorem}

The proof has two algorithmic ingredients.  The first is the sampler for
\(2\)-in/\(2\)-out directed Eulerian graphs, where one local transition flip is
a single bit.

\begin{theorem}[Degree-two sampler]\label{thm:degree-two-sampler}
Let \(H\) be a strongly connected directed Eulerian multigraph with \(M\) arcs
and \(\deg^+(v)=\deg^-(v)=2\) for every vertex.  For every \(0<\eps<1\), there
is a randomized algorithm which outputs a tour whose distribution is within
\(\eps\) total variation distance of the uniform distribution on \(\ET(H)\).
Its running time is
\[
  O\!\left(M^{3/2}\polylog(M/\eps)\right),
\]
and its space usage is \(O(M)\).
\end{theorem}

The second ingredient is a standard reduction from arbitrary degrees to degree
two.  The reduction replaces a degree-\(D\) vertex by a \(D\)-wire switching
network: \(D\) input wires enter a small directed gadget of two-state
\(2\)-in/\(2\)-out switches, \(D\) output wires leave it, and a setting of all
switches routes the inputs to the outputs, thereby realizing a permutation in
\(S_D\).  Networks that realize permutations by composing binary switches are
classical in the theory of rearrangeable connecting and permutation networks
\cite{Benes1964,Waksman1968}; randomized versions have also been studied as
ways to generate almost-random permutations
\cite{Morris2009Thorp,Morris2009Improved,Czumaj2015SwitchingNetworks}.  For
our reduction, the network needs pointwise multiplicative control on every
local permutation.  This is stronger than merely saying that the induced
permutation is almost uniform in total variation distance.  We use the
following pointwise \(L_\infty\) form as an auxiliary tool.

\begin{theorem}[Random pointwise switching networks]\label{thm:pointwise-switching}
For every \(D\ge2\), \(0<\eta<1\), and \(0<\delta<1\), there is a randomized
algorithm which constructs a \(D\)-wire degree-two switching network with
\[
  O\!\left(D\,\polylog(D/(\eta\delta))\right)
\]
switches and arcs.  With probability at least \(1-\delta\) over the construction
of the network, the following pointwise guarantee holds.  If all switches are
then set uniformly and independently and \(\mu_{\cN}\) is the induced
distribution on \(S_D\), then
\[
  e^{-\eta}\frac1{D!}
  \le
  \mu_{\cN}(\sigma)
  \le
  e^\eta\frac1{D!}
  \qquad\forall\sigma\in S_D.
\]
\end{theorem}

The runtime bound is driven by the degree-two sampler in
\cref{thm:degree-two-sampler}: the chain mixes in nearly linear many steps and
the data structure implements each step in \(O(\sqrt M\,\polylog M)\) time.
The extension from arbitrary degrees to degree two is a supporting wrapper
around this core.  The one feature we need from that wrapper is pointwise
control: common total-variation formulations of almost-random switching
networks do not by themselves imply the multiplicative lift-counting estimate
used here.

This places the paper in a broader algorithmic theme.  Several recent
breakthroughs in combinatorial optimization obtain near-linear running times by
combining an optimization primitive with a data structure that maintains the
primitive under many local updates; the almost-linear time algorithms for
maximum flow and minimum-cost flow are a prominent example
\cite{ChenKyngLiuPengProbstGutenbergSachdeva2022}.  The analogous story has
been rarer in sampling.  Random spanning tree generation is one important
exception: the near-linear sampler of Anari, Liu, Oveis Gharan, Vinzant, and
Vuong combines sharp mixing bounds for down-up walks with graph-algorithmic
data structures that implement the walk efficiently
\cite{AnariLiuOveisGharanVinzantVuong2021}.  The present paper gives another
example of this sampling-side paradigm: the speedup comes from matching a new
local Markov chain to a representation in which its updates can be maintained
dynamically.

The algorithm has three pieces, shown schematically in \cref{fig:overview}.
The degree-two flip--repair chain is the core of the paper.  The expansion and
projection steps are the interface between this core sampler and general
Eulerian digraphs: they convert high-degree local choices into degree-two
switches, run the degree-two sampler, and then project the sampled tour back to
the original graph.

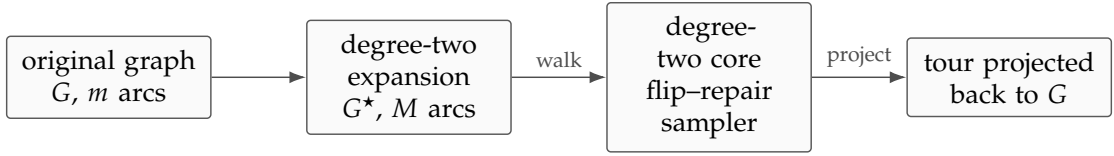
\begin{figure}[t]
  \centering
  \begin{tikzpicture}[
    >=Latex,
    every node/.style={font=\small},
    object/.style={
      draw=black!65,
      line width=.55pt,
      rounded corners=1.5pt,
      fill=black!2,
      align=center,
      text width=2.35cm,
      minimum height=1cm,
      inner sep=5pt
    },
    arrow/.style={-{Latex[length=2.2mm,width=1.6mm]}, line width=.55pt, black!70}
  ]
    \node[object] (input) {original graph\\\(G\), \(m\) arcs};
    \node[object, right=1.25cm of input] (expand) {degree-two expansion\\\(G^\star\), \(M\) arcs};
    \node[object, right=1.25cm of expand] (walk) {degree-two core\\flip--repair sampler};
    \node[object, right=1.25cm of walk] (output) {tour projected\\back to \(G\)};

    \draw[arrow] (input) -- (expand);
    \draw[arrow] (expand) -- node[above, font=\scriptsize, text=black!70] {walk} (walk);
    \draw[arrow] (walk) -- node[above, font=\scriptsize, text=black!70] {project} (output);
  \end{tikzpicture}
  \caption{The sampler reduces to the degree-two flip--repair walk and then
  projects back to \(G\).  Pointwise switching networks make the projection
  close to uniform, with \(M=\card{E(G^\star)}=O(m\,\polylog(m/\eps))\).}
  \label{fig:overview}
\end{figure}

\paragraph{Degree two.}
When every vertex has indegree and outdegree two, each vertex has exactly two
possible local transitions.  Fix a reference Eulerian tour \(C\).  At a vertex
\(v\), one transition is the pairing used by \(C\); the other swaps the two
successors.  Thus every transition system is obtained by flipping a subset
\(S\subseteq V\) of vertices relative to \(C\).  The Cohn--Lempel
interlacement theorem \cite{CohnLempel1972,Traldi2011}, together with
Bouchet's unimodular orientation theorem for circle graphs
\cite{Bouchet1987Unimodularity,Bouchet1992UnimodularOrientations,
BouchetCunninghamGeelen1998}, gives a real skew-symmetric matrix
\(A=A_C\) such that
\[
  \det(A_{S,S})\in\{0,1\},
  \qquad
  \det(A_{S,S})=1
  \Longleftrightarrow
  C^S \text{ is one Eulerian tour}.
\]
Thus the uniform distribution on degree-two tours is the
\emph{skew-determinantal measure}
\[
  \mu_A(S)\propto \det(A_{S,S})
\]
for a real skew-symmetric matrix \(A\).  We call such a measure \emph{flat}
when all positive determinants are equal; in the degree-two tour case they are
all \(1\).  This name emphasizes that the probability weights are principal
determinants of a skew-symmetric matrix, and no additional structure is meant
by the terminology.  Since every odd skew-symmetric principal determinant is
zero, this also explains why all feasible flips lie on one parity layer.

The degree-two sampler uses the flip--repair walk on this support.  From a
feasible \(S\), choose \(i\in V\) uniformly and move to the odd set
\(Y=S\symdiff\{i\}\).  Then repair by choosing uniformly among all feasible
\(T\) with \(\card{T\symdiff Y}=1\).  In tour language, the first flip splits
the current tour into two cycles.  The repair flip is chosen uniformly among the
local flips that merge the two cycles back into one; this set includes flipping
\(i\) back, which gives the self-loop \(T=S\).

\paragraph{Mixing.}
The main proof ingredient behind the runtime is a mixing theorem for the larger
family of skew-determinantal measures.  If \(A\) is any real skew-symmetric
\(n\times n\) matrix, define
\[
  w_A(S)=\det(A_{S,S}),
  \qquad
  \mu_A(S)=\frac{w_A(S)}{\sum_T w_A(T)}.
\]
By a basic linear-algebra fact for real skew-symmetric matrices
\cite{HornJohnson2012}, odd principal determinants are zero and even principal
determinants are nonnegative.  The formula has the same principal-minor shape as
determinant-based sampling measures, but the matrix is skew-symmetric and the
support lies on one parity class.  The same even-to-odd flip--repair walk has
stationary distribution \(\mu_A\).  We prove:

\begin{theorem}[Skew-determinantal flip--repair mixing]\label{thm:pfaffian-intro}
For every real skew-symmetric \(n\times n\) matrix \(A\), with \(n\ge2\), the
skew-determinantal flip--repair walk has spectral gap at least \(2/n\).  If, in
addition, \(w_A(S)\in\{0,1\}\) for every \(S\subseteq[n]\), then
\[
  \Ent_{\mu_A}(f^2)
  \le
  n(\log n+2)\,\cE_P(f,f)
  \qquad\text{for every }f:\supp(\mu_A)\to\R.
\]
Consequently, in the flat case, for every \(0<\eps<1\),
\[
  t_{\mathrm{mix}}(\eps)
  \le
  O\!\left(n\log n\,[1+\log n+\log(1/\eps)]\right)
\]
from every starting state in the support.
\end{theorem}

The spectral-gap statement is the structural part of the theorem.  The constant
\(2/n\) is sharp in the following simple case: \(A\) is a direct sum of
\(2\times2\) skew-symmetric blocks, so the support is a product of independent
binary choices, and the flip--repair walk updates one block with probability
\(2/n\).  Thus no skew-symmetric matrix can make the parity flip--repair walk
slower than this product case; the proof of \cref{thm:spectral-comparison}
spells out the sharpness example.  For the sampling algorithm, however, the gap
alone would only give a mixing bound with an extra factor coming from
\(\log(1/\min_S\mu_A(S))\), which can be linear in \(n\) even in the flat case.
The log-Sobolev part is what turns the flat \(0/1\) structure of Eulerian tours
into nearly linear mixing time, up to logarithmic factors.

The proof of the mixing theorem has a useful linear-algebraic summary.  For an
even principal submatrix of a skew-symmetric matrix, the determinant has a
standard signed square root, the Pfaffian.  We put one basis vector \(e_S\) for
each subset \(S\subseteq[n]\), and collect all these signed square roots in the
single vector
\[
  \Psi_A=\sum_{S\text{ even}}\pf(A_{S,S})e_S
\]
inside a \(2^n\)-dimensional vector space.  Signed insertion and deletion
operators act on the subset coordinate \(S\).  Combining them gives the signed
coordinate flips that model the first half-step of the chain.  Applying these
signed flips to the normalized vector \(\Psi_A/\norm{\Psi_A}\) gives
orthonormal excitation coordinates, with \(\Psi_A/\norm{\Psi_A}\) as the
degree-zero vector.  We prove this from the standard \(2\times2\) block normal
form for real skew-symmetric matrices and covariance of exterior powers under
orthogonal changes of coordinates \cite{HornJohnson2012}.  A row-cancellation
estimate for the same coordinates follows from Knuth's overlapping-Pfaffian
identity \cite{Knuth1996}.  In these excitation coordinates, the chain's
Dirichlet form dominates the excitation number divided by \(n\), giving the
sharp gap \(2/n\).

For flat measures, this spectral estimate is strengthened by an entropy
argument.  The state basis and the excitation basis have maximum overlap
\(1/\sqrt{\card{\supp(\mu_A)}}\).  Maassen--Uffink entropic uncertainty
\cite{MaassenUffink1988}, followed by Rothaus's centering lemma
\cite{Rothaus1981}, converts this overlap bound into the stated log-Sobolev
inequality.  The entire mixing proof is independent of Eulerian tours and is
stated for all real skew-symmetric matrices.

\paragraph{Implementation and degree reduction.}
The mixing theorem gives the number of Markov-chain steps; we also need to
implement each step quickly.  In degree two, represent the current tour by
placing its \(M\) arc occurrences in cyclic order on a circle.  Each vertex has
two incoming arc occurrences in this order, and joining these two positions by a
chord gives the cyclic chord diagram of the current tour.  After the first flip
at \(x\), the valid repair flips are \(x\) itself and the vertices whose chords
cross \(x\).
For a degree-two graph with \(M\) arcs, we maintain chunks of size on the order
of \(\sqrt M\), pair lists for chunk pairs, and a top-level balanced
sequence tree with aggregate counts.  This supports uniform sampling of the
repair flip and the resulting
four-cut tour update in \(O(\sqrt M\,\polylog M)\) amortized time.
The log-Sobolev bound gives \(O(M\polylog(M/\eps))\) flip--repair steps, so
this is where the \(M^{3/2}\), and ultimately \(m^{3/2}\), running time comes
from.

To pass from degree two to arbitrary degrees, we first suppress degree-one
vertices, whose local transition is forced.  Each remaining degree-\(d\) vertex
is then replaced by a degree-two switching-network gadget on \(d\) wires.  This
is a standard local reduction rather than the source of the new running time.
Closely related gadgets appear in prior work on Eulerian-tour and A-trail
reductions, and in parallel sampling reductions for Eulerian tours
\cite{GeStefankovic2012,AnariHuangLiuVuongXuYu2023}.  We use it here because it
lets the degree-two sampler be applied without changing the distribution by
much.

The only point that needs care is the strength of the approximation.  We use a
sparse random network of lazy transposition switches.  With high probability
over the random choice of the network, the uniform distribution on switch
settings induces a distribution on \(S_d\) which is pointwise
\(e^{\pm\eta}\)-close to uniform.  This is the multiplicative form needed for
lift counting; total-variation closeness of the local permutation distribution
would not by itself be enough, because the analysis conditions on the global
one-cycle event.  The pointwise estimate is a quenched switching-network version
of the random-transposition shuffle bounds of Diaconis and Shahshahani
\cite{DiaconisShahshahani1981}.  The lift-counting argument then shows that
every original tour has nearly the same number of expanded tours, point by
point.  Choosing local errors and gadget-failure probabilities summing to
\(O(\eps)\) makes the projected output distribution \(O(\eps)\)-close to uniform
on \(\ET(G)\).

\paragraph{Related work.}
The classical exact approach to directed Eulerian tour sampling is through the
BEST theorem and random arborescences
\cite{KandelMatiasUngerWinkler1996,ProppWilson1998,JiangEtAl2008}.  Creed and
Cryan studied Eulerian tours in random directed regular graphs
\cite{CreedCryan2013}.  Tetali and Vempala studied a Kotzig-transformation chain
for the undirected Eulerian-tour sampling problem \cite{TetaliVempala1997}.  That
problem is different from the directed problem considered here, and a general
efficient sampler for undirected Euler tours remains open.

The walk studied here is on directed transition systems and does not use the
undirected Kotzig chain.  After reducing to degree two, the state space is a
flat skew-determinantal support: it is described by principal determinants of a
real skew-symmetric matrix, and those determinants are all \(0\) or \(1\).  The
proof proceeds by linear algebra of signed square roots of these principal
determinants rather than comparison to arborescences.

Switching networks and the related Thorp shuffle have been studied as
oblivious ways to generate almost-random permutations
\cite{Morris2009Thorp,Morris2009Improved,Czumaj2015SwitchingNetworks}.  Local
random permutation gadgets have also been used in reductions involving
Eulerian tours and A-trails \cite{GeStefankovic2012}, and in parallel sampling
work for Eulerian tours \cite{AnariHuangLiuVuongXuYu2023}.  Our use is a
supporting reduction: the degree-two sampler and its skew-determinantal mixing
theorem do the main work, while the gadget supplies pointwise multiplicative
control before conditioning on the global one-cycle event.

The measure \(\mu_A\) is a skew-symmetric principal-minor measure: its
coefficients are \(\det(A_{S,S})\), where \(A\) is real skew-symmetric.  These
coefficients are nonnegative because
\(\det(A_{S,S})=\pf(A_{S,S})^2\), and they vanish on one parity.  Thus
\(\mu_A\) is close in form to determinantal and strongly Rayleigh measures.
Strongly Rayleigh measures, introduced by Borcea, Branden, and Liggett
\cite{BorceaBrandenLiggett2009}, satisfy strong negative dependence properties
and have rapidly mixing basis-exchange walks in the homogeneous case
\cite{AnariOveisGharanRezaei2016}.  Those results do not directly cover the
measure here: the principal-minor formula comes from a skew-symmetric matrix,
the support lies on one parity class rather than on one cardinality layer, and
the natural local move is flip--repair rather than basis exchange.  The
generating polynomial
\[
  \sum_S \det(A_{S,S})z^S
  =
  \det(I+A\diag(z))
\]
is Hurwitz stable: if \(\Re z_i>0\) for all \(i\), then
\(\diag(z)^{-1}+A\) has positive definite Hermitian part and is nonsingular.
We formulate conjectures in \cref{sec:hurwitz} suggesting that the spectral and
stochastic-covering phenomena proved here for skew-determinantal measures extend
to general Hurwitz-stable distributions; stochastic covering refers to local
couplings between neighboring one-coordinate conditionings.  This is also close
to recent stability-based approaches to spectral independence and rapid mixing
\cite{AnariAlimohammadiShiragurVuong2021,ChenLiuVigoda2021,ChenFengYinZhang2024}.

\paragraph{Provenance and AI assistance.}
The author's contribution was the high-level algorithmic plan, the
transition-system and gadget reduction, the dynamic chord data structure, and
the conjecture that the skew-determinantal flip--repair walk has the stated
mixing behavior.  The linear-algebra proof of \cref{thm:pfaffian-intro} was
produced by GPT 5.5 Pro Extended.  Codex then organized, edited, and typeset the
manuscript.  See \cref{app:provenance} for a fuller note.

\paragraph{Organization.}
\Cref{sec:prelim} fixes transition-system notation.  \Cref{sec:interlace}
identifies degree-two tours with flat skew-determinantal measures.
\Cref{sec:pfaffian-walk} proves the general skew-determinantal mixing theorem.
\Cref{sec:data-structure} gives the dynamic implementation of the degree-two
sampler, and \cref{sec:gadgets} proves the pointwise degree-reduction step.
\Cref{sec:hurwitz} records the stochastic-covering side result and the
Hurwitz-stable conjectures.

\section{Preliminaries}\label{sec:prelim}

\paragraph{Eulerian transition systems.}
Let \(G=(V,E)\) be a directed Eulerian multigraph.  A transition system
\(\tau\) assigns to each \(v\in V\) a bijection
\[
  \tau_v:\delta^-(v)\to\delta^+(v).
\]
For a loop, its incoming and outgoing incidences are distinct local incidences,
even though they belong to the same arc.
It induces a permutation \(\pi_\tau\) of \(E\) by
\[
  \pi_\tau(e)=\tau_{\operatorname{head}(e)}(e).
\]
The cycles of \(\pi_\tau\) are directed circuits partitioning \(E\).  Thus
\(\tau\) is an Eulerian tour if and only if \(\pi_\tau\) has one cycle.  We
write \(\ET(G)\) for the set of one-cycle transition systems.  A rooted,
linearly written tour is obtained by choosing one of the \(\card{E}\) arcs as
the first arc, so this convention changes counts by a factor of \(\card{E}\).

\paragraph{The degree-two flip--repair chain.}
Assume \(\deg^+(v)=\deg^-(v)=2\) for every vertex under consideration.  Fix a
reference tour \(C\).  Each vertex has two possible transitions.  If the two
incoming incidences of \(v\) are paired by \(C\) with outgoing incidences
\(o_1,o_2\), then the crossed transition swaps these two outgoing incidences.
Therefore transition systems are indexed by subsets \(S\subseteq V\), where
\(C^S\) denotes the system obtained by flipping exactly \(S\).  The degree-two
flip--repair chain on feasible \(S\)'s is:
\[
  S \xrightarrow{\text{choose } i\in V \text{ uniformly}} Y=S\symdiff\{i\}
  \xrightarrow{\text{choose feasible }T,\ \card{T\symdiff Y}=1\text{ uniformly}} T.
\]
The feasible choices always include flipping \(i\) back to return to \(S\);
thus the chain has a self-loop at every state.

\paragraph{Entropy and Dirichlet forms.}
For a probability measure \(\mu\) on a finite set and a nonnegative function \(g\),
\[
  \Ent_\mu(g)=\E_\mu{g\log g}-\E_\mu{g}\log\E_\mu{g}.
\]
We use the convention \(0\log0=0\).
For a reversible Markov kernel \(P\) with stationary distribution \(\mu\),
\[
  \cE_P(f,f)=\ip{f}{(I-P)f}_\mu
  =
  \frac12\sum_{x,y}\mu(x)P(x,y)(f(x)-f(y))^2.
\]

\paragraph{Pfaffians.}
For an even-dimensional skew-symmetric matrix \(B\) with ordered rows and
columns, \(\pf(B)\) denotes the signed square root of the determinant specified
by the perfect-matching expansion; equivalently,
\(\pf(B)^2=\det(B)\).  We use the conventions
\(\pf(\varnothing)=1\) and \(\pf(B)=0\) for odd-dimensional skew-symmetric
matrices.  The proof only needs this determinant identity, expansion along a
row or column, and Knuth's overlapping-Pfaffian identity \cite{Knuth1996},
which is cited again where it is used.

\section{Degree-Two Tours as Flat Skew-Determinantal Measures}\label{sec:interlace}

We next identify the degree-two state space with a principal-minor support.
The input is a fixed reference tour \(C\).  Relative to \(C\), each vertex
contributes one bit, recording whether its local transition is flipped.  The
classical interlacement description of circuit partitions then turns the
one-cycle condition into a principal-minor test
\cite{CohnLempel1972,Traldi2011}.  This is the bridge from Eulerian tours to
the skew-determinantal measures analyzed in the next section.

Let \(H\) be a \(2\)-in/\(2\)-out directed Eulerian multigraph, and let \(C\) be
an Eulerian tour.  For each vertex \(v\), mark the two visits of \(C\) at which an
incoming incidence of \(v\) is followed by an outgoing incidence of \(v\).  A
loop contributes distinct incoming and outgoing incidences at the same vertex,
so it is treated in this incidence sense.  These two marked positions form a
chord on the tour circle.  Two vertices \(u,v\) are interlaced if their four
positions alternate around the circle.

The unsigned interlacement matrix \(I_C\) over \(\mathbb F_2\) has entry \(1\)
when two chords cross and diagonal entries \(0\).  The Cohn--Lempel equality
\cite{CohnLempel1972,Traldi2011} says that if \(S\) is the set of vertices at
which the transition is flipped relative to \(C\), then the number of circuits
in the resulting circuit partition is
\[
  1+\operatorname{nullity}_{\mathbb F_2}(I_C[S,S]).
\]
In particular, \(C^S\) is an Eulerian tour if and only if \(I_C[S,S]\) is
nonsingular over \(\mathbb F_2\).  The principal matrix \(I_C[S,S]\) is
symmetric with zero diagonal.  Over \(\mathbb F_2\), this means that the
bilinear form \(x^\top I_C[S,S]y\) is alternating, i.e.,
\[
  x^\top I_C[S,S]x=0
  \qquad\text{for every }x\in\mathbb F_2^S.
\]
Every alternating matrix has even rank, by the usual symplectic elimination for
alternating bilinear forms.  Thus, when \(\card{S}\) is odd, \(I_C[S,S]\) has
rank strictly smaller than \(\card{S}\) and is singular.  The feasible sets
therefore lie on the even layer, matching the parity of the skew-determinantal
representation.

For the mixing theorem we need a real skew-symmetric lift of this binary test.
We now spell out the signing.  Fix an orientation of the tour circle, and choose
one endpoint \(v^-\) of every chord \(v\) as its tail; write \(v^+\) for the
other endpoint.  For distinct vertices \(u,v\), define
\[
  A_C(u,v)=
  \begin{cases}
    1,
      & \text{if \(u\) and \(v\) cross and \(v^-\) lies on the open arc
        from \(u^-\) to \(u^+\)},\\
    -1,
      & \text{if \(u\) and \(v\) cross and \(v^+\) lies on the open arc
        from \(u^-\) to \(u^+\)},\\
    0,
      & \text{if \(u\) and \(v\) do not cross,}
  \end{cases}
\]
and set \(A_C(v,v)=0\).  When \(u\) and \(v\) cross, exactly one endpoint of
\(v\) lies on the open arc from \(u^-\) to \(u^+\), so the definition is
unambiguous; it also gives \(A_C(v,u)=-A_C(u,v)\).  Reducing entries modulo
\(2\) forgets the signs and gives \(A_C\bmod2=I_C\).

Changing the chosen tail of a chord \(v\) multiplies the \(v\)-th row and column
of \(A_C\) by \(-1\), so all principal determinants are independent of these
auxiliary choices.  Bouchet's theorem says that this signed adjacency matrix of
a circle graph is principally unimodular; the later paper of Bouchet,
Cunningham, and Geelen uses essentially this formulation
\cite{Bouchet1987Unimodularity,Bouchet1992UnimodularOrientations,
BouchetCunninghamGeelen1998}.

\begin{theorem}[Bouchet's unimodular signing of a circle graph]\label{thm:circle-pu}
For every chord diagram, the matrix \(A_C\) defined above satisfies
\[
  \det(A_C[S,S])\in\{0,1\}
  \qquad\text{for every }S\subseteq V.
\]
Moreover \(A_C\bmod2=I_C\).
\end{theorem}

\begin{proof}
The congruence \(A_C\bmod2=I_C\) is immediate from the definition.  The
principal determinant statement is Bouchet's unimodularity theorem for circle
graphs \cite{Bouchet1987Unimodularity}.  The same chord-orientation construction
is stated explicitly in the circle-graph discussion of Bouchet, Cunningham, and
Geelen \cite{BouchetCunninghamGeelen1998}.  Their terminology allows principal
determinants \(0\) or \(\pm1\).  In the present skew-symmetric case,
\(\det(A_C[S,S])=\pf(A_C[S,S])^2\), so the determinants are nonnegative and
hence lie in \(\{0,1\}\).
\end{proof}

\begin{theorem}[Flat skew-determinantal representation]\label{thm:tour-pfaffian}
Let \(H\) be a \(2\)-in/\(2\)-out directed Eulerian multigraph and fix an
Eulerian tour \(C\).  There is a real skew-symmetric matrix
\(A_C\in\R^{V\times V}\) such that for every \(S\subseteq V\),
\[
  \det(A_C[S,S])\in\{0,1\},
  \qquad
  \det(A_C[S,S])=1
  \Longleftrightarrow
  C^S\in\ET(H).
\]
Consequently, the uniform distribution on \(\ET(H)\), represented relative to
\(C\), is
\[
  \mu_{A_C}(S)
  =
  \frac{\det(A_C[S,S])}{\sum_T \det(A_C[T,T])}.
\]
\end{theorem}

\begin{proof}
The Cohn--Lempel equality gives the equivalence between being a single circuit
and nonsingularity of the binary interlacement principal submatrix
\cite{CohnLempel1972,Traldi2011}.  Let \(A_C\) be the oriented interlacement
matrix defined above.  By \cref{thm:circle-pu}, each principal determinant of
\(A_C\) is \(0\) or \(1\).

It remains to match the real test with the binary one.  Since \(A_C\bmod2=I_C\),
we have
\[
  \det(A_C[S,S])\bmod2=\det(I_C[S,S])\quad\text{in }\mathbb F_2.
\]
For a determinant already known to be \(0\) or \(1\), being nonzero over the
reals is therefore equivalent to being nonsingular over \(\mathbb F_2\).
Combining this with Cohn--Lempel proves the theorem.  The denominator is nonzero
because the reference tour \(C=C^\varnothing\) is feasible and
\(\det(A_C[\varnothing,\varnothing])=1\).
\end{proof}

For this particular matrix \(A_C\), every positive weight is equal to \(1\).
Consequently the normalizing constant is exactly the number of feasible
degree-two tours represented relative to \(C\), and each feasible subset is
counted once.  This is the sense in which the distribution is flat.
Thus the abstract skew-determinantal flip--repair walk has a concrete tour
interpretation in degree two.  After flipping one vertex \(i\), the repair step
chooses uniformly among the vertices \(j\) for which flipping \(j\) gives
another one-cycle transition system.  The choice \(j=i\) is always available
and is the self-loop.  This is exactly the degree-two flip--repair chain
defined in \cref{sec:prelim}.  The next section proves that this walk mixes
rapidly, in fact for every real skew-symmetric matrix, not only for matrices
arising from chord diagrams.

\section{The Skew-Determinantal Flip--Repair Walk}\label{sec:pfaffian-walk}

We now prove the mixing theorem used by the sampler.  The Eulerian-tour
matrices from \cref{sec:interlace} are \(0/1\)-weighted examples of the
following more general object, so we state the argument at that level.
Let \(A\in\R^{n\times n}\) be skew-symmetric.  For \(S\subseteq[n]\), write
\(A_S=A_{S,S}\), and set \(\pf(A_S)=0\) when \(\card{S}\) is odd.  Define
\[
  w(S)=\det(A_S)=\pf(A_S)^2,
  \qquad
  Z=\sum_{S\subseteq[n]}w(S).
\]
We use the convention that the empty principal determinant is \(1\), so
\(Z\ge1\).  The equality \(w(S)=\pf(A_S)^2\) also explains why these weights
are nonnegative.
Since odd skew-symmetric matrices have determinant zero, \(w\) is supported on
even sets.  Let
\[
  \Omega=\{S\subseteq[n]:w(S)>0\},
  \qquad
  \mu(S)=w(S)/Z.
\]
We call \(\mu\) the skew-determinantal measure associated with \(A\).

The flip--repair walk first flips one coordinate and then applies the repair
step, which is the time reversal of that flip.  We write this two-step chain as
\(P=KK^*\), using the following \(L^2\) convention: \(K(S,Y)\) is the forward
transition probability from even sets to odd sets, while as an operator on
functions \(K\) averages a function on the odd layer back to the even layer.
Its adjoint \(K^*\), with respect to \(\mu\) and the intermediate measure
\(\nu\), is the repair step from odd sets to even sets.  The forward kernel
sends an even \(S\in\Omega\) to an odd set \(Y\) by flipping a uniformly random
coordinate:
\[
  K(S,Y)=\frac1n\1_{\{\card{S\symdiff Y}=1\}}.
\]
The intermediate measure is
\[
  \nu(Y)=\sum_S\mu(S)K(S,Y)
  =
  \frac1{nZ}\sum_{i=1}^n w(Y\symdiff\{i\}).
\]
For odd \(Y\) with \(\nu(Y)>0\), the reverse kernel is the conditional law of
the previous even state given that the first flip reached \(Y\):
\[
  K^*(Y,T)
  =
  \frac{w(T)}
  {\sum_i w(Y\symdiff\{i\})}
  \qquad
  \text{if }T\in\Omega\text{ and }\card{T\symdiff Y}=1,
\]
and \(K^*(Y,T)=0\) otherwise.  The denominator sums the weights of all even
neighbors of \(Y\); neighbors of weight zero simply contribute nothing.
Thus
\[
  P(S,T)=
  \frac1n
  \sum_{\substack{Y:\ \card{Y\symdiff S}=\card{Y\symdiff T}=1}}
  \frac{w(T)}{\sum_i w(Y\symdiff\{i\})}.
\]
Here \(S,T\in\Omega\).  Every odd \(Y\) that appears in the sum is reachable
from \(S\), so the denominator is positive because it includes the term
\(w(S)\).
The two layers and the repair step are illustrated in
\cref{fig:pfaffian-walk}.
The chain is reversible with stationary distribution \(\mu\): with the
operator convention above, \(P=KK^*\) is self-adjoint in \(L^2(\mu)\), and
Markov kernels of this form preserve constants.
When \(\Omega\) has one state, we use the convention that the spectral gap is
\(1\) and the mixing time is \(0\).

\begin{figure}[t]
  \centering
  \begin{tikzpicture}[
    >=Latex,
    every node/.style={font=\small},
    state/.style={circle, draw=black!65, fill=black!2, minimum size=9mm, inner sep=1pt},
    defect/.style={circle, draw=teal!65!black, fill=teal!7, minimum size=9mm, inner sep=1pt},
    arr/.style={-{Latex[length=2.1mm,width=1.5mm]}, line width=.55pt, black!70},
    band/.style={rounded corners=2pt, fill=black!3, draw=black!18, line width=.45pt},
    defectband/.style={rounded corners=2pt, fill=teal!5, draw=teal!45!black, line width=.45pt}
  ]
    \path[band] (-.82,-2.06) rectangle (.82,2.22);
    \path[defectband] (1.98,-2.06) rectangle (3.62,2.22);
    \path[band] (4.93,-2.06) rectangle (6.57,2.22);

    \node[state] (S) at (0,0) {\(S\)};
    \node[defect] (Y1) at (2.8,.75) {\(Y\)};
    \node[defect] (Y2) at (2.8,-.75) {\(Y'\)};
    \node[state] (T1) at (5.75,1.25) {\(T_1\)};
    \node[state] (T2) at (5.75,.35) {\(T_2\)};
    \node[state] (T3) at (5.75,-.55) {\(T_3\)};
    \node[state] (T4) at (5.75,-1.45) {\(T_4\)};

    \draw[arr] (S) -- node[above, sloped, font=\scriptsize] {flip \(i\)} (Y1);
    \draw[arr] (S) -- node[below, sloped, font=\scriptsize] {flip \(i'\)} (Y2);
    \draw[arr] (Y1) -- node[above, sloped, font=\scriptsize] {\(K^*\)} (T1);
    \draw[arr] (Y1) -- (T2);
    \draw[arr] (Y2) -- (T3);
    \draw[arr] (Y2) -- node[below, sloped, font=\scriptsize] {\(K^*\)} (T4);

    \node[font=\scriptsize, text=black!65] at (0,1.98) {even};
    \node[font=\scriptsize, text=teal!65!black] at (2.8,1.98) {odd defect};
    \node[font=\scriptsize, text=black!65] at (5.75,1.98) {even};
  \end{tikzpicture}
  \caption{One step of the skew-determinantal flip--repair walk.  From an even
  feasible set \(S\), the chain flips a uniformly random coordinate to reach an
  odd intermediate set.  It then samples from the feasible even neighbors of
  that odd set with probabilities proportional to \(w(T)=\det(A_{T,T})\).}
  \label{fig:pfaffian-walk}
\end{figure}
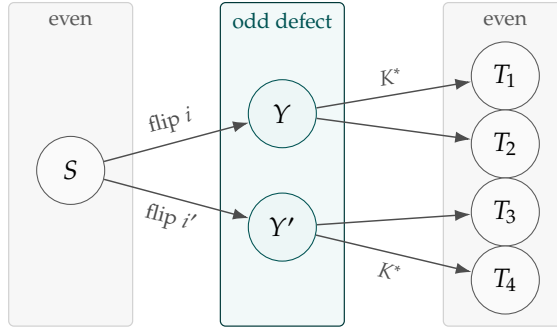

\begin{theorem}[Skew-determinantal spectral gap]\label{thm:spectral-comparison}
For every real skew-symmetric \(n\times n\) matrix \(A\), with \(n\ge2\), the
skew-determinantal flip--repair walk satisfies
\[
  \operatorname{gap}(P)\ge \frac2n.
\]
The constant \(2/n\) is sharp.
\end{theorem}

\begin{theorem}[Flat log-Sobolev inequality]\label{thm:flat-lsi}
Assume \(n\ge2\) and \(w(S)\in\{0,1\}\) for all \(S\subseteq[n]\).  Then
\[
  \Ent_\mu(f^2)
  \le
  n(\log n+2)\,\cE_P(f,f)
  \qquad\text{for every }f:\Omega\to\R.
\]
Consequently, for every \(0<\eps<1\),
\[
  t_{\mathrm{mix}}(\eps)
  \le
  O\!\left(n\log n[1+\log n+\log(1/\eps)]\right).
\]
The bound is from every starting state in \(\Omega\).
\end{theorem}

The proof is a comparison argument after a square-root change of variables.
Instead of analyzing the transition matrix directly on functions
\(f:\Omega\to\R\), we multiply each coordinate by a signed square root of its
weight:
\[
  h_f=\sum_{S\in\Omega} f(S)\pf(A_S)e_S .
\]
The point of this normalization is that the repair step becomes an orthogonal
projection.  The second step is to use the basis obtained by applying signed
coordinate flips to the normalized Pfaffian vector.  We call the resulting
basis vectors excitations, by analogy with Fourier characters on the Boolean
cube.  The excitation degree is measured by a diagonal operator
\(\mathsf N_A\).

The central estimate compares the projection operator of the walk with this
degree operator:
\[
  \mathsf P_A\preceq I-\frac1n\mathsf N_A .
\]
Once this is proved, the spectral gap is almost immediate.  The degree-zero
excitation is the constant direction, and every other even excitation has degree
at least \(2\).  Thus the Dirichlet form is at least \(2/n\) times the variance.

The flat log-Sobolev bound uses the same coordinates.  When all nonzero weights
are \(1\), no state basis vector has a large overlap with any excitation basis
vector.  Entropic uncertainty turns this incoherence into an entropy bound.  In
this order, the rest of the section constructs \(\mathsf P_A\), constructs
\(\mathsf N_A\), compares them, and then applies the same coordinates to
entropy.

\subsection{The Square-Root Operator}

We first rewrite the walk in the square-root coordinates above.  Fix an odd
intermediate set \(Y\).  The repair step averages over the feasible even
neighbors of \(Y\), with weights \(w(S)\).  After replacing \(w(S)\) by
\(\pf(A_S)^2\), this weighted average becomes an orthogonal projection onto a
single vector \(\Phi_Y\).  Averaging over the first flip gives the operator
\(\mathsf P_A\).

Let \(V=\R^n\), with standard basis \(e_1,\ldots,e_n\), and let
\[
  \cH=\bigwedge\nolimits^\bullet V
  =
  \bigoplus_{k=0}^n \bigwedge\nolimits^k V .
\]
For \(S=\{i_1<\cdots<i_k\}\), write
\[
  e_S=e_{i_1}\wedge\cdots\wedge e_{i_k},
  \qquad e_\varnothing=1.
\]
The vectors \(e_S\) form an orthonormal basis.  We use
\(\cH^{\mathrm{even}}\) and \(\cH^{\mathrm{odd}}\) for the spans of the even
and odd basis vectors.

The Pfaffian square-root vector of \(A\) is the vector whose coordinates are the
principal Pfaffians:
\[
  \Psi=\Psi_A=
  \sum_{S\text{ even}} \pf(A_S)e_S .
\]
Equivalently, by the perfect-matching expansion of the Pfaffian
\cite{Knuth1996},
\[
  \Psi=
  \exp\left(\sum_{1\le i<j\le n} A_{ij}e_i\wedge e_j\right).
\]
Therefore \(\norm{\Psi}^2=Z\).  For a function \(f:\Omega\to\R\), put
\[
  h_f=f\Psi
  :=
  \sum_{S\in\Omega} f(S)\pf(A_S)e_S .
\]
Then
\begin{equation}\label{eq:sqrt-embedding}
  \norm{h_f}^2=Z\E_\mu{f^2},
  \qquad
  \ip{h_f}{\Psi}=Z\E_\mu{f}.
\end{equation}

For an odd set \(Y\), let
\[
  \Star(Y)=\set{S\subseteq[n]\given \card{S\symdiff Y}=1\text{ and }S\text{ is even}},
  \qquad
  C(Y)=\sum_{S\in \Star(Y)} w(S).
\]
Define
\[
  \Phi_Y=
  \sum_{S\in \Star(Y)}\pf(A_S)e_S\in\cH^{\mathrm{even}} .
\]
Thus \(\norm{\Phi_Y}^2=C(Y)\).  Let
\(\operatorname{Proj}_u\) denote orthogonal projection onto the line spanned by
a nonzero vector \(u\).  Define the square-root operator
\[
  \mathsf P_A
  =
  \frac1n
  \sum_{\substack{Y\text{ odd}\\ C(Y)>0}}
  \operatorname{Proj}_{\Phi_Y}
  \qquad\text{on }\cH^{\mathrm{even}} .
\]

\begin{lemma}[Square-root representation]\label{lem:square-root-representation}
For every \(f:\Omega\to\R\),
\[
  \ip{f}{Pf}_\mu
  =
  \frac1Z\ip{h_f}{\mathsf P_A h_f}.
\]
\end{lemma}

\begin{proof}
For an odd \(Y\),
\[
  \ip{h_f}{\Phi_Y}
  =
  \sum_{S\in \Star(Y)} f(S)\pf(A_S)^2
  =
  \sum_{S\in \Star(Y)} f(S)w(S).
\]
The odd marginal of the first flip is \(\nu(Y)=C(Y)/(nZ)\), and the repair
average at \(Y\) is
\[
  (K^*f)(Y)=
  \frac{\sum_{S\in \Star(Y)} f(S)w(S)}{C(Y)}
\]
when \(C(Y)>0\).  Hence
\[
  \ip{f}{Pf}_\mu
  =
  \norm{K^*f}_{L^2(\nu)}^2
  =
  \frac1{nZ}
  \sum_{\substack{Y\text{ odd}\\ C(Y)>0}}
  \frac{\ip{h_f}{\Phi_Y}^2}{\norm{\Phi_Y}^2},
\]
which is exactly \(Z^{-1}\ip{h_f}{\mathsf P_Ah_f}\).
\end{proof}

\subsection{Excitation Coordinates}

We now build the basis in which the comparison is diagonal.  The analogy is
with the Fourier basis on the Boolean cube.  There the constant function is the
degree-zero vector, and multiplying by coordinate characters produces basis
vectors whose degree is the number of coordinates used.  Here the normalized
Pfaffian vector \(\widehat\Psi=\Psi/\sqrt Z\) plays the role of the constant
vector.  Applying signed coordinate flips to \(\widehat\Psi\) gives the other
basis vectors, and the number of flips applied is the excitation degree.

The signed flips are the standard insertion-plus-deletion operators on subset
vectors.  For \(b\in\R^n\), let
\[
  \varepsilon_b x=b\wedge x,
  \qquad
  \iota_b=\varepsilon_b^*,
  \qquad
  c_b=\varepsilon_b+\iota_b .
\]
Thus \(\varepsilon_b\) inserts \(b\), \(\iota_b\) deletes \(b\), and
\(c_b\) is the signed coordinate flip.  The exterior-algebra relations are
\[
  \varepsilon_b\varepsilon_{b'}+\varepsilon_{b'}\varepsilon_b=0,
  \qquad
  \iota_b\iota_{b'}+\iota_{b'}\iota_b=0,
  \qquad
  \iota_b\varepsilon_{b'}+\varepsilon_{b'}\iota_b=\ip{b}{b'}I.
\]
Equivalently,
\[
  c_bc_{b'}+c_{b'}c_b=2\ip{b}{b'}I .
\]
These are the standard finite-dimensional creation--annihilation relations
\cite{Berezin1966,Bravyi2005}.  We write \(c_i=c_{e_i}\), and for
\(T=\{i_1<\cdots<i_k\}\),
\[
  c_T=c_{i_1}\cdots c_{i_k}.
\]

Set \(\widehat\Psi=\Psi/\sqrt Z\).  The excitation vectors are
\[
  \xi_T=c_T\widehat\Psi,
  \qquad T\subseteq[n].
\]

\begin{lemma}[Excitation basis]\label{lem:excitation-basis}
The vectors \(\{\xi_T:T\subseteq[n]\}\) form an orthonormal basis of \(\cH\).
The vectors with \(T\) even form an orthonormal basis of
\(\cH^{\mathrm{even}}\).
\end{lemma}

\begin{proof}
We first check orthonormality in the canonical block form, where the calculation
factors into two-dimensional pieces, and then reduce the general case to this
one by an orthogonal change of coordinates.

In the block case, \(A\) is a direct sum of \(2\times2\) skew blocks
\(\begin{psmallmatrix}0&a_r\\-a_r&0\end{psmallmatrix}\), with possibly one
additional zero row and column.  Fix one block \(\{p,q\}\), write \(a=a_r\),
and set \(r=\sqrt{1+a^2}\).  The normalized factor of
\(\widehat\Psi\) on this block is
\[
  u_0=\frac{1+a\,e_p\wedge e_q}{r}.
\]
Using the displayed anticommutation relations, or equivalently the definitions
of insertion and deletion, the four vectors obtained from \(u_0\) are
\[
\begin{aligned}
  u_0
  &=
  \frac{1+a\,e_p\wedge e_q}{r},
  &
  c_pu_0
  &=
  \frac{e_p+a e_q}{r},
  \\
  c_qu_0
  &=
  \frac{-a e_p+e_q}{r},
  &
  c_pc_qu_0
  &=
  \frac{-a+e_p\wedge e_q}{r}.
\end{aligned}
\]
Each has norm one.  The two even vectors are orthogonal because
\(\ip{1+a\,e_p\wedge e_q}{-a+e_p\wedge e_q}=0\), and the two odd vectors are
orthogonal because \(\ip{e_p+a e_q}{-a e_p+e_q}=0\).  Even and odd vectors are
orthogonal to each other.  Thus \(u_0,c_pu_0,c_qu_0,c_pc_qu_0\) form an
orthonormal basis of \(\bigwedge^\bullet\R^{\{p,q\}}\).  Tensoring these bases
over all blocks, and using the usual basis on any zero coordinate, proves the
block case.

For a general skew-symmetric \(A\), choose an orthogonal matrix \(Q\) such that
\(B=QAQ^\mathsf T\) is in this block form \cite{HornJohnson2012}.  Let
\(\Lambda(Q)\) be the induced action of \(Q\) on exterior powers:
\[
  \Lambda(Q)(v_1\wedge\cdots\wedge v_k)
  =
  Qv_1\wedge\cdots\wedge Qv_k .
\]
By Cauchy--Binet, \(\Lambda(Q)\) is orthogonal.  It satisfies
\[
  \Lambda(Q)\varepsilon_b=\varepsilon_{Qb}\Lambda(Q),
  \qquad
  \Lambda(Q)\iota_b=\iota_{Qb}\Lambda(Q),
  \qquad
  \Lambda(Q)c_b=c_{Qb}\Lambda(Q).
\]
The Pfaffian vector is covariant under the same change of coordinates:
\[
  \Lambda(Q)\Psi_A=\Psi_{QAQ^\mathsf T}=\Psi_B.
\]
In coordinates this is the Pfaffian Cauchy--Binet, or minor-summation,
identity \cite{Knuth1996}.  In the present notation it follows directly from
the exterior exponential formula for \(\Psi_A\).  Let
\[
  \omega_A=\sum_{i<j}A_{ij}e_i\wedge e_j
  =
  \frac12\sum_{i,j}A_{ij}e_i\wedge e_j .
\]
Since \(\Lambda(Q)\) respects wedge products,
\[
\begin{aligned}
  \Lambda(Q)\omega_A
  &=
  \frac12\sum_{i,j}A_{ij}(Qe_i)\wedge(Qe_j)  \\
  &=
  \frac12\sum_{r,s}(QAQ^\mathsf T)_{rs}e_r\wedge e_s
  =
  \omega_{QAQ^\mathsf T}.
\end{aligned}
\]
Therefore
\[
  \Lambda(Q)\Psi_A
  =
  \Lambda(Q)\exp(\omega_A)
  =
  \exp(\Lambda(Q)\omega_A)
  =
  \exp(\omega_{QAQ^\mathsf T})
  =
  \Psi_{QAQ^\mathsf T}.
\]

Write \(\xi_T^A\) and \(\widehat\Psi_B\) for the excitation vector defined from
\(A\) and the normalized Pfaffian vector defined from \(B\).  Then
\[
  \Lambda(Q)\xi_T^A
  =
  c_{Qe_{i_1}}\cdots c_{Qe_{i_k}}\widehat\Psi_B .
\]
Let \(g_j=c_{Qe_j}\).  Since the columns of \(Q\) are orthonormal, the
operators \(g_j\) satisfy the same relations as the \(c_j\)'s:
\[
  g_jg_\ell+g_\ell g_j=2\delta_{j\ell}I.
\]
Thus, for distinct \(i_1,\ldots,i_k\), all contraction terms vanish in
\(g_{i_1}\cdots g_{i_k}\): it is the signed-toggle image of the exterior
product
\[
  (Qe_{i_1})\wedge\cdots\wedge(Qe_{i_k})
\]
under the map \(e_U\mapsto c_U\).  Expanding this exterior product gives
\[
  c_{Qe_{i_1}}\cdots c_{Qe_{i_k}}
  =
  \sum_{\substack{U\subseteq[n]\\ \card{U}=k}}
  \det(Q_{U,T})\,c_U .
\]
Thus the Gram matrix of the vectors
\(c_{Qe_{i_1}}\cdots c_{Qe_{i_k}}\widehat\Psi_B\) is the compound-matrix
Gram matrix
\[
  \sum_{\substack{U\subseteq[n]\\ \card{U}=k}}
  \det(Q_{U,T})\det(Q_{U,T'}),
\]
which is \(\delta_{T,T'}\) by Cauchy--Binet.  Orthogonality for \(A\) follows
by applying \(\Lambda(Q)^*\).  Since each \(c_i\) changes parity, \(\xi_T\)
lies in \(\cH^{\mathrm{even}}\) exactly when \(T\) is even.
\end{proof}

Define the excitation number operator by
\[
  \mathsf N_A\xi_T=\card{T}\,\xi_T.
\]
For \(f:\Omega\to\R\), define its average excitation degree by
\[
  \Deg_A(f)
  =
  \frac1Z\ip{h_f}{\mathsf N_A h_f}.
\]
Thus \(\mathsf N_A\) is simply the degree operator in the excitation basis.  This
is the usual number operator in finite fermionic Fock space after the orthogonal
change of generators determined by \(\widehat\Psi\)
\cite{Berezin1966,Bravyi2005}, but the proof uses only the displayed diagonal
form.

\subsection{Projection versus Number}

We now prove the central inequality
\(\mathsf P_A\preceq I-\mathsf N_A/n\).  We need to understand each projection
vector \(\Phi_Y\) in excitation coordinates.

Write an even vector \(h\in\cH^{\mathrm{even}}\) in excitation coordinates as
\[
  h=\sqrt Z\sum_{T\text{ even}} a_T\xi_T,
  \qquad
  a=\sum_{T\text{ even}} a_Te_T .
\]
For an odd set \(Y\), define its coordinate row in the excitation basis by
\[
  R_Y=
  \sum_{T\text{ odd}} \ip{e_Y}{\xi_T}e_T
  \in\bigwedge\nolimits^{\mathrm{odd}}\R^n.
\]
The vectors \(R_Y\), over odd \(Y\), form an orthonormal basis of the odd
coordinate space; \(R_Y\) is the row of the change-of-basis matrix indexed by
the ordinary odd basis vector \(e_Y\).  Let
\[
  \beta_Y=\pi_1R_Y\in\bigwedge\nolimits^1\R^n\simeq\R^n,
\]
where \(\pi_1\) is orthogonal projection onto degree one.
The point of \(\beta_Y\) is that this degree-one part contains exactly the
information needed to express the projection at \(Y\).

\begin{lemma}[Projection-coordinate identity]\label{lem:projection-coordinate}
For every odd \(Y\),
\[
  C(Y)=Z\norm{\beta_Y}^2.
\]
Moreover, if \(h=\sqrt Z\sum_Ta_T\xi_T\), then
\[
  \ip{h}{\Phi_Y}=Z\ip{R_Y}{c_{\beta_Y}a}.
\]
Consequently, when \(C(Y)>0\) and
\(\widehat\beta_Y=\beta_Y/\norm{\beta_Y}\),
\[
  \frac1Z\frac{\ip{h}{\Phi_Y}^2}{\norm{\Phi_Y}^2}
  =
  \abs*{\ip{R_Y}{c_{\widehat\beta_Y}a}}^2 .
\]
\end{lemma}

\begin{proof}
The vector \(R_Y\) is the row of the change-of-basis matrix from excitation
coordinates to ordinary odd subset coordinates.  Its degree-one part
\(\beta_Y\) has one coordinate for each even neighbor \(Y\symdiff\{i\}\) of
\(Y\), with the corresponding Pfaffian sign.
For each coordinate \(i\),
\[
  (\beta_Y)_i
  =
  \ip{R_Y}{e_i}
  =
  \ip{e_Y}{c_i\widehat\Psi}
  =
  \frac{(c_i\Psi)_Y}{\sqrt Z}.
\]
For fixed \(Y\) and \(i\), let \(s(Y,i)\in\{\pm1\}\) be the exterior sign with
\[
  (c_i x)_Y=s(Y,i)x_{Y\symdiff\{i\}}
\]
for every vector \(x\) in the subset basis.  Therefore
\[
  (c_i\Psi)_Y^2=w(Y\symdiff\{i\}),
\]
and summing over \(i\) gives \(C(Y)=Z\norm{\beta_Y}^2\).

The same sign convention applies in excitation coordinates: \(c_i\xi_T\)
changes the index \(T\) with the same exterior sign as \(c_i e_T\) changes the
ordinary basis vector \(e_T\).  Therefore \(c_i a\) is the
excitation-coordinate vector of \(c_i h/\sqrt Z\).  Hence
\[
  \ip{R_Y}{c_i a}
  =
  \frac{(c_i h)_Y}{\sqrt Z}.
\]
It follows that
\[
\begin{aligned}
  \ip{R_Y}{c_{\beta_Y}a}
  &=
  \sum_i (\beta_Y)_i\ip{R_Y}{c_i a} \\
  &=
  \frac1Z\sum_i (c_i\Psi)_Y(c_i h)_Y \\
  &=
  \frac1Z\sum_{S\in \Star(Y)} \pf(A_S)h_S \\
  &=
  \frac1Z\ip{h}{\Phi_Y}.
\end{aligned}
\]
The third line uses the sign notation above: if \(S=Y\symdiff\{i\}\), then
\((c_i\Psi)_Y(c_i h)_Y=s(Y,i)^2\pf(A_S)h_S=\pf(A_S)h_S\).  The normalized
formula follows from \(\norm{\Phi_Y}^2=C(Y)=Z\norm{\beta_Y}^2\).
\end{proof}

The next lemma is the only cancellation needed for the comparison.  It says
that each row \(R_Y\) is killed by insertion in the direction of its degree-one
part \(\beta_Y\).  The cancellation is the following specialization of
Knuth's overlapping-Pfaffian relation \cite{Knuth1996}.  We write the needed
special case explicitly in the proof, after the signs have been fixed.

\begin{lemma}[Row isotropy]\label{lem:row-isotropy}
For every odd set \(Y\),
\[
  \varepsilon_{\beta_Y}R_Y=0.
\]
\end{lemma}

\begin{proof}
For odd \(T\), define the sign
\[
  \chi_Y(T)=\ip{e_Y}{c_Te_{Y\symdiff T}}\in\{\pm1\}.
\]
Since \(\xi_T=c_T\Psi/\sqrt Z\),
\begin{equation}\label{eq:row-coordinate}
  \ip{R_Y}{e_T}
  =
  \ip{e_Y}{\xi_T}
  =
  \frac1{\sqrt Z}\chi_Y(T)\pf(A_{Y\symdiff T}).
\end{equation}
Write
\[
  \operatorname{ins}(i,U)
  =
  (-1)^{\card{\set{j\in U\given j<i}}}
\]
for the sign in \(\varepsilon_i e_U=\operatorname{ins}(i,U)e_{U\cup\{i\}}\).

Fix an even set \(R\).  The coefficient of \(e_R\) in
\(\varepsilon_{\beta_Y}R_Y\) is
\[
  \sum_{i\in R}
  \operatorname{ins}(i,R\setminus\{i\})
  \ip{R_Y}{e_i}
  \ip{R_Y}{e_{R\setminus\{i\}}}.
\]
By \cref{eq:row-coordinate}, this is \(1/Z\) times
\[
  \sum_{i\in R}
  \sigma_Y(R,i)
  \pf(A_{Y\symdiff\{i\}})
  \pf(A_{Y\symdiff(R\setminus\{i\})}),
\]
where
\[
  \sigma_Y(R,i)
  =
  \operatorname{ins}(i,R\setminus\{i\})
  \chi_Y(\{i\})
  \chi_Y(R\setminus\{i\}).
\]
Knuth's overlapping-Pfaffian identity, specialized to the two ordered lists
\(Y\symdiff R\) and \(R\) and written in this ordered-subset convention, gives
\begin{equation}\label{eq:knuth-specialized}
  \sum_{i\in R}
  \sigma_Y(R,i)
  \pf(A_{Y\symdiff\{i\}})
  \pf(A_{Y\symdiff(R\setminus\{i\})})
  =0.
\end{equation}
Elements common to the two lists are retained in both lists; this is exactly
the overlapping case in Knuth's relation.  The sign \(\sigma_Y(R,i)\) is the
sign obtained by moving the singled-out element \(i\) through the ordered subset
basis, once for the insertion \(e_{R\setminus\{i\}}\mapsto e_R\) and once for
each of the two toggle coefficients in \(\chi_Y\).
The sum is zero by the specialized overlapping-Pfaffian identity
\cref{eq:knuth-specialized}.  Hence every coefficient of
\(\varepsilon_{\beta_Y}R_Y\) is zero.
\end{proof}

\begin{lemma}[Number domination]\label{lem:number-domination}
As quadratic forms on \(\cH^{\mathrm{even}}\),
\[
  \mathsf P_A\preceq I-\frac1n\mathsf N_A.
\]
\end{lemma}

\begin{proof}
At this point the proof is a Bessel-type estimate.  The projection-coordinate
identity expresses each rank-one projection in \(\mathsf P_A\) as the square of
an inner product against \(c_{\widehat\beta_Y}a\).  Row isotropy removes the
deletion part of \(c_{\widehat\beta_Y}\), leaving only insertions.  Summing over
\(Y\) then uses orthonormality of the rows \(R_Y\).

Let \(h=\sqrt Z\sum_Ta_T\xi_T\) be even.  By the definition of
\(\mathsf P_A\) and \cref{lem:projection-coordinate},
\begin{equation}\label{eq:projection-number-start}
  \frac1Z\ip{h}{\mathsf P_Ah}
  =
  \frac1n
  \sum_{\substack{Y\text{ odd}\\ \beta_Y\ne0}}
  \abs*{\ip{R_Y}{c_{\widehat\beta_Y}a}}^2 .
\end{equation}
For such a \(Y\), row isotropy gives
\(\varepsilon_{\widehat\beta_Y}R_Y=0\).  Since
\(\iota_{\widehat\beta_Y}^*=\varepsilon_{\widehat\beta_Y}\),
\[
  \ip{R_Y}{\iota_{\widehat\beta_Y}a}
  =
  \ip{\varepsilon_{\widehat\beta_Y}R_Y}{a}
  =0.
\]
Thus \(c_{\widehat\beta_Y}\) may be replaced by
\(\varepsilon_{\widehat\beta_Y}\) in \cref{eq:projection-number-start}.  By
Cauchy--Schwarz in the coordinate \(i\),
\[
  \abs*{\ip{R_Y}{\varepsilon_{\widehat\beta_Y}a}}^2
  =
  \abs*{\sum_i (\widehat\beta_Y)_i\ip{R_Y}{\varepsilon_i a}}^2
  \le
  \sum_i\abs{\ip{R_Y}{\varepsilon_i a}}^2 .
\]
Summing over \(Y\) and using that the \(R_Y\)'s form an orthonormal basis of the
odd coordinate space gives
\[
  \frac1Z\ip{h}{\mathsf P_Ah}
  \le
  \frac1n\sum_i\norm{\varepsilon_i a}^2.
\]
If \(N_0e_T=\card{T}e_T\) is the ordinary degree operator, then
\[
  \sum_i\varepsilon_i^*\varepsilon_i
  =
  \sum_i\iota_i\varepsilon_i
  =
  nI-N_0,
\]
because \(\iota_i\varepsilon_i\) is the identity on basis vectors not
containing \(i\) and is zero on basis vectors containing \(i\).  Hence
\[
  \frac1Z\ip{h}{\mathsf P_Ah}
  \le
  \frac1n\ip{a}{(nI-N_0)a}
  =
  \frac1Z\ip*{h}{\left(I-\frac1n\mathsf N_A\right)h}.
\]
This proves the quadratic-form domination.
\end{proof}

\begin{lemma}[Degree comparison]\label{lem:degree-comparison}
For every \(f:\Omega\to\R\),
\[
  \cE_P(f,f)\ge \frac1n\Deg_A(f).
\]
\end{lemma}

\begin{proof}
Apply \cref{lem:square-root-representation,lem:number-domination} to
\(h_f\):
\[
  \ip{f}{Pf}_\mu
  \le
  \frac1Z\ip*{h_f}{\left(I-\frac1n\mathsf N_A\right)h_f}
  =
  \E_\mu{f^2}-\frac1n\Deg_A(f).
\]
Subtracting from \(\E_\mu{f^2}\) gives the claim.
\end{proof}

\begin{proof}[Proof of \cref{thm:spectral-comparison}]
We now return from the subset-vector space to the Markov chain.  The operator
comparison gives a lower bound on the Dirichlet form in terms of excitation
degree; the parity of the state space supplies the factor \(2\).

Assume first that \(\card{\Omega}\ge2\).  If \(\E_\mu{f}=0\), then
\Cref{eq:sqrt-embedding} gives \(h_f\perp\Psi\).  The vector
\(\xi_\varnothing=\widehat\Psi\) is the unique excitation vector of degree zero.
Also \(h_f\) lies in the even subspace, so all remaining excitation degrees are
at least \(2\).  Therefore
\[
  \Deg_A(f)
  \ge
  2\norm{f}_{L^2(\mu)}^2
  =
  2\Var_\mu(f).
\]
Together with \cref{lem:degree-comparison}, this gives
\(\cE_P(f,f)\ge(2/n)\Var_\mu(f)\), which is the spectral-gap bound.  It also
implies irreducibility on \(\Omega\): otherwise a nonconstant function that is
constant on each communicating class would have zero Dirichlet form.

If \(\card{\Omega}=1\), our convention gives gap \(1\), and the theorem is
immediate for \(n\ge2\).

Sharpness follows from a block-diagonal example.  Let \(A\) be a direct sum of
\(2\times2\) skew blocks, with one additional zero coordinate when \(n\) is odd.
The support is a product of \(\lfloor n/2\rfloor\) binary choices.  A
flip--repair step resamples a given block with probability \(2/n\), and the
second eigenvalue is exactly \(1-2/n\).  Thus the gap is exactly \(2/n\).
\end{proof}

\subsection{Flat Log-Sobolev Bound}

Assume from now on that \(n\ge2\) and \(w(S)\in\{0,1\}\) for all
\(S\subseteq[n]\).  Then \(Z=\card{\Omega}\).  The spectral proof did not use
flatness.  The flat assumption enters through an incoherence bound between the
state basis \(\{e_S\}\) and the excitation basis \(\{\xi_T\}\).
This subsection follows the standard entropy route: incoherence gives entropic
uncertainty, entropic uncertainty bounds entropy by excitation-basis entropy,
and a one-line Gibbs bound controls that entropy by the average excitation
degree.  Rothaus's centering lemma removes the mean.

\begin{lemma}[Flat incoherence]\label{lem:flat-incoherence}
Assume \(w(S)\in\{0,1\}\) for every \(S\).  For even \(S,T\),
\[
  \abs{\ip{e_S}{\xi_T}}\le \frac1{\sqrt Z}.
\]
\end{lemma}

\begin{proof}
Since \(\xi_T=c_T\Psi/\sqrt Z\), the operator \(c_T\) moves the coefficient at
\(S\symdiff T\) to the coefficient at \(S\), up to sign.  Hence
\[
  \ip{e_S}{\xi_T}
  =
  \pm\frac{\pf(A_{S\symdiff T})}{\sqrt Z}.
\]
If \(S\symdiff T\notin\Omega\), this number is zero.  If
\(S\symdiff T\in\Omega\), flatness gives
\(\det(A_{S\symdiff T})=1\), and the Pfaffian has absolute value \(1\).
\end{proof}

\begin{lemma}[State-excitation uncertainty]\label{lem:state-excitation-uncertainty}
Assume \(w(S)\in\{0,1\}\) for every \(S\).  Let
\(\norm{f}_{L^2(\mu)}=1\), set \(v=h_f/\sqrt Z\), and define
\[
  p_S=\abs{\ip{e_S}{v}}^2,
  \qquad
  q_T=\abs{\ip{\xi_T}{v}}^2 .
\]
Then
\[
  \Ent_\mu(f^2)\le H(q),
\]
where \(H(q)=-\sum_Tq_T\log q_T\).
\end{lemma}

\begin{proof}
The distributions \(p\) and \(q\) are the squared coordinates of the same unit
vector in two orthonormal bases of \(\cH^{\mathrm{even}}\).  Since \(h_f\) is
supported on \(\Omega\) and \(Z=\card{\Omega}\),
\[
  p_S=
  \begin{cases}
    f(S)^2/Z,& S\in\Omega,\\
    0,& S\notin\Omega.
  \end{cases}
\]
Therefore
\[
  H(p)=\log Z-\Ent_\mu(f^2).
\]
The Maassen--Uffink entropic uncertainty inequality \cite{MaassenUffink1988}
says that if two orthonormal bases have maximum overlap \(c\), then
\[
  H(p)+H(q)\ge -2\log c.
\]
By \cref{lem:flat-incoherence}, \(c\le1/\sqrt Z\), so
\(H(p)+H(q)\ge\log Z\).  Combining this with the displayed formula for \(H(p)\)
gives the claim.
\end{proof}

\begin{lemma}[Centered entropy]\label{lem:centered-entropy}
Assume \(w(S)\in\{0,1\}\) for every \(S\).  If \(\E_\mu{f}=0\), then
\[
  \Ent_\mu(f^2)\le (\log n+1)\Deg_A(f).
\]
\end{lemma}

\begin{proof}
If \(f=0\), the claim is immediate.  By homogeneity, assume
\(\norm{f}_{L^2(\mu)}=1\).  Set \(v=h_f/\sqrt Z\) and
\(q_T=\abs{\ip{\xi_T}{v}}^2\).  By
\cref{lem:state-excitation-uncertainty},
\[
  \Ent_\mu(f^2)\le H(q).
\]
Since \(\E_\mu{f}=0\), we have \(h_f\perp\Psi\), so \(q_\varnothing=0\).  Also
\(v\) lies in the even subspace, so \(q_T=0\) for odd \(T\), and
\[
  \Deg_A(f)=\sum_{T\text{ even}} \card{T}q_T.
\]
For any \(\theta>0\), the Gibbs variational principle gives
\[
  H(q)
  \le
  \theta\sum_T \card{T}q_T
  +
  \log\sum_{\substack{T\text{ even}\\T\ne\varnothing}} e^{-\theta\card{T}}.
\]
Indeed, this is the nonnegativity of relative entropy from \(q\) to the
probability distribution proportional to \(e^{-\theta\card{T}}\) on nonempty
even sets.  Taking \(\theta=\log n\),
\[
  \sum_{\substack{T\text{ even}\\T\ne\varnothing}} n^{-\card{T}}
  \le
  \sum_{T\ne\varnothing} n^{-\card{T}}
  =
  (1+1/n)^n-1
  \le e-1.
\]
All mass of \(q\) is on nonempty even sets, so \(\Deg_A(f)\ge2\).  Since
\(\log(e-1)<1\),
\[
  H(q)
  \le
  (\log n)\Deg_A(f)+\log(e-1)
  \le
  (\log n+1)\Deg_A(f).
\]
\end{proof}

\begin{lemma}[Rothaus centering]\label{lem:rothaus}
For every real \(a\) and every real function \(g\),
\[
  \Ent_\mu((g+a)^2)
  \le
  \Ent_\mu(g^2)+2\E_\mu{g^2}.
\]
\end{lemma}

\begin{proof}
This is the centering inequality of Rothaus \cite{Rothaus1981}.  We include the
standard differentiable calculation in the nonvanishing case to fix the
constant used here; the finite-state inequality follows by approximation.
Assume first that \((tg+a)^2\) is positive on \(\Omega\) for \(0<t<1\).  Let
\[
  G_a(t)=\Ent_\mu((tg+a)^2),
  \qquad
  R(t)=t^2\Ent_\mu(g^2)+2t^2\E_\mu{g^2}.
\]
Then \(G_a(0)=R(0)=0\) and \(G_a'(0)=R'(0)=0\).  Direct differentiation gives
\[
  G_a''(t)
  =
  2\E*_\mu{
    g^2\log
    \frac{(tg+a)^2}{\E_\mu{(tg+a)^2}}
  }
  +
  4\E_\mu{g^2}
  -
  \frac{4\E_\mu{g(tg+a)}^2}{\E_\mu{(tg+a)^2}}.
\]
The last term is nonpositive.  If
\[
  r=\frac{(tg+a)^2}{\E_\mu{(tg+a)^2}},
\]
then \(\E_\mu{r}=1\), and the variational formula for entropy gives
\[
  \E_\mu{g^2\log r}\le \Ent_\mu(g^2).
\]
Thus \(G_a''(t)\le R''(t)\).  Integrating twice from \(0\) to \(1\) gives the
claim.  The general finite-state case follows by replacing \((g,a)\) with a
small perturbation for which the nonvanishing assumption holds and taking a
limit.
\end{proof}

\begin{proof}[Proof of \cref{thm:flat-lsi}]
If \(\card{\Omega}=1\), every function is constant and the claim is immediate.
Assume \(\card{\Omega}\ge2\).  Let \(g=f-\E_\mu{f}\).  By
\cref{lem:rothaus},
\[
  \Ent_\mu(f^2)
  \le
  \Ent_\mu(g^2)+2\Var_\mu(f).
\]
The centered entropy lemma gives
\[
  \Ent_\mu(g^2)\le(\log n+1)\Deg_A(g).
\]
Subtracting the mean changes only the \(\xi_\varnothing\) coefficient, so
\(\Deg_A(g)=\Deg_A(f)\).  Since all nonconstant even excitation degrees are at
least \(2\),
\[
  \Var_\mu(f)\le\frac12\Deg_A(f).
\]
Therefore
\[
  \Ent_\mu(f^2)
  \le
  (\log n+2)\Deg_A(f).
\]
Using \cref{lem:degree-comparison},
\[
  \Ent_\mu(f^2)
  \le
  n(\log n+2)\cE_P(f,f).
\]

Let \(C_{\mathrm{LS}}=n(\log n+2)\).  By \cref{thm:spectral-comparison}, the
chain is irreducible on \(\Omega\).  The log-Sobolev-to-mixing theorem for
finite reversible chains, in the convention
\(\Ent_\mu(f^2)\le C_{\mathrm{LS}}\cE_P(f,f)\), gives for the continuous-time
semigroup \(e^{-t(I-P)}\)
\cite{DiaconisSaloffCoste1996,LevinPeresWilmer2017}
\[
  t_{\mathrm{mix}}^{\mathrm{cont}}(x,\eps)
  \le
  O\left(
    C_{\mathrm{LS}}
    \left[
      1+\log\log\frac1{\mu(x)}+
      \log\frac1\eps
    \right]
  \right).
\]
In the flat case \(\mu\) is uniform on \(\Omega\subseteq2^{[n]}\), so
\(1+\log\log(1/\mu(x))\le O(1+\log n)\).

Finally, \(P=KK^*\) is positive semidefinite, so every eigenvalue of \(P\) lies
in \([0,1]\).  On each nonconstant eigenvector, \(\lambda^t\le e^{-t(1-\lambda)}\),
while the stationary eigenvector is unchanged.  Thus the discrete-time chain is
no slower in \(L^2(\mu)\) than the continuous-time chain.  The standard
\(L^2(\mu)\)-to-total-variation comparison transfers the same total-variation
bound to discrete time.
\end{proof}

\section{Implementing the Degree-Two Chain}\label{sec:data-structure}

The previous section bounds the number of flip--repair steps.  We now describe
how to execute those steps in sublinear time on a changing tour.  Let \(H\) be
a \(2\)-in/\(2\)-out Eulerian digraph with \(M\) arcs.  The fixed reference tour
used in \cref{sec:interlace} is needed for the mixing proof, but not for the
implementation.  The data structure stores only the current one-cycle
transition system, written as a cyclic sequence
\[
  \Gamma=(e_1,e_2,\ldots,e_M).
\]
An occurrence means one position of this cyclic sequence.  Since each position
is an arc whose head is some vertex, every occurrence is owned by the head of
that arc; in a \(2\)-in/\(2\)-out graph each vertex owns exactly two
occurrences.  Flipping a coordinate of the subset representation is exactly
swapping the two successors after the two occurrences owned by the same vertex
in the current transition system.  For a vertex \(v\), let \(a_v,b_v\) be these
two incoming-arc occurrences, viewed as positions on the cyclic sequence
\(\Gamma\).  Thus every vertex is a chord of the current tour circle.

\begin{lemma}[Crossing criterion]\label{lem:crossing}
Let \(x\) be the first flipped vertex.  Its two occurrences split the tour
circle into two open intervals \(I\) and \(I^c\).  After flipping \(x\), the
tour becomes the two directed cycles corresponding to these intervals.  A
repair vertex \(y\ne x\) restores a single Eulerian tour if and only if exactly
one of \(a_y,b_y\) lies in \(I\).  Equivalently, the chords \(x\) and \(y\)
cross.
In addition, flipping \(x\) back restores the original tour.
\end{lemma}

\begin{proof}
Flipping \(x\) swaps the successors after the two incoming-arc occurrences
\(a_x,b_x\).  In cyclic order, this cuts the original directed cycle at the two
positions and reconnects each interval to itself, giving two directed cycles.  A
repair swap at a vertex \(y\ne x\) merges two cycles exactly when its two
owned occurrences lie on different cycles.  These are precisely the chords
crossing \(x\).  The remaining feasible repair flip is \(y=x\), which undoes the
first flip.
\end{proof}

Thus the repair step, in the flat skew-determinantal tour measure, is exactly
dynamic uniform sampling from the queried chord itself and all chords crossing
it in the current tour.  Sampling the queried chord gives a self-loop; sampling
a crossing chord is followed by a constant number of cyclic cuts and
concatenations.  \Cref{fig:crossing-move} shows both the crossing test and the
resulting cyclic reassembly.
The implementation therefore never needs to materialize the odd intermediate
two-cycle state.  It samples the first vertex \(x\), samples the repair vertex
from \(\{x\}\) together with the chords crossing \(x\) in the current tour, and
then either stays put or applies the combined two-flip update.
The naive implementation scans all chords after the first flip, costing
\(O(M)\) per flip--repair step.  Since the mixing bound already has \(O(M)\)
steps up to polylogarithmic factors, this would give a quadratic-time sampler.
We use a chunk structure to answer the same crossing query in roughly
\(\sqrt M\) time while still supporting the four-cut update to the tour.

\begin{figure}[t]
  \centering
  \begin{tikzpicture}[
    every node/.style={font=\small},
    >=Latex,
    dot/.style={circle, draw=black!70, fill=white, inner sep=1.25pt},
    chordx/.style={line width=1pt, blue!65!black},
    chordy/.style={line width=1pt, teal!60!black},
    blockp/.style={line width=2.4pt, blue!55!black, line cap=round},
    blockq/.style={line width=2.4pt, teal!60!black, line cap=round},
    blockr/.style={line width=2.4pt, orange!75!black, line cap=round},
    blocks/.style={line width=2.4pt, black!55, line cap=round},
    arr/.style={-{Latex[length=2.2mm,width=1.5mm]}, line width=.55pt, black!70}
  ]
    \begin{scope}[xshift=-.2cm]
      \node[font=\small] at (0,1.78) {crossing criterion};
      \draw[blue!12, line width=6pt, line cap=round] (145:1.18) arc (145:-25:1.18);
      \draw[black!10, line width=6pt, line cap=round] (-25:1.18) arc (-25:-215:1.18);
      \draw[black!55, line width=.55pt] (0,0) circle (1.18);
      \coordinate (a) at (145:1.18);
      \coordinate (c) at (55:1.18);
      \coordinate (b) at (-25:1.18);
      \coordinate (d) at (-135:1.18);
      \draw[chordx] (a) -- (b);
      \draw[chordy] (c) -- (d);
      \node[dot, label=left:\(a\)] at (a) {};
      \node[dot, label=above:\(c\)] at (c) {};
      \node[dot, label=right:\(b\)] at (b) {};
      \node[dot, label=below:\(d\)] at (d) {};
      \node[blue!65!black, font=\scriptsize] at (.18,.35) {\(x\)};
      \node[teal!60!black, font=\scriptsize] at (-.28,-.34) {\(y\)};
      \node[font=\scriptsize, text=blue!65!black] at (.82,.74) {\(I\)};
      \node[font=\scriptsize, text=black!55] at (-.9,-.72) {\(I^c\)};
      \node[draw=black!35, rounded corners=1.5pt, fill=white, align=center,
        font=\scriptsize, inner sep=2pt] at (0,-1.62) {flip \(x\) back\\gives \(S\)};
    \end{scope}

    \begin{scope}[xshift=3.35cm]
      \node[font=\small] at (2.65,1.78) {cyclic reassembly};
      \node[font=\scriptsize, anchor=east, text=black!65] at (-.25,.55) {before};
      \node[font=\scriptsize, anchor=east, text=black!65] at (-.25,-.85) {after};

      \node[dot] at (0,.55) {\(\scriptstyle a\)};
      \draw[blockp] (.28,.55) -- node[midway, above=2pt, font=\scriptsize] {\(P\)} (1.12,.55);
      \node[dot] at (1.4,.55) {\(\scriptstyle c\)};
      \draw[blockq] (1.68,.55) -- node[midway, above=2pt, font=\scriptsize] {\(Q\)} (2.52,.55);
      \node[dot] at (2.8,.55) {\(\scriptstyle b\)};
      \draw[blockr] (3.08,.55) -- node[midway, above=2pt, font=\scriptsize] {\(R\)} (3.92,.55);
      \node[dot] at (4.2,.55) {\(\scriptstyle d\)};
      \draw[blocks] (4.48,.55) -- node[midway, above=2pt, font=\scriptsize] {\(S\)} (5.32,.55);

      \draw[arr] (2.65,.18) -- node[right, font=\scriptsize, align=left, text=black!65] {flip \(x\), then \(y\)} (2.65,-.48);

      \node[dot] at (0,-.85) {\(\scriptstyle a\)};
      \draw[blockr] (.28,-.85) -- node[midway, below=2pt, font=\scriptsize] {\(R\)} (1.12,-.85);
      \node[dot] at (1.4,-.85) {\(\scriptstyle d\)};
      \draw[blockq] (1.68,-.85) -- node[midway, below=2pt, font=\scriptsize] {\(Q\)} (2.52,-.85);
      \node[dot] at (2.8,-.85) {\(\scriptstyle b\)};
      \draw[blockp] (3.08,-.85) -- node[midway, below=2pt, font=\scriptsize] {\(P\)} (3.92,-.85);
      \node[dot] at (4.2,-.85) {\(\scriptstyle c\)};
      \draw[blocks] (4.48,-.85) -- node[midway, below=2pt, font=\scriptsize] {\(S\)} (5.32,-.85);
    \end{scope}
  \end{tikzpicture}
  \caption{In a \(2\)-in/\(2\)-out graph, the first flip cuts the current tour
  at the two occurrences of \(x\).  Apart from flipping \(x\) back, a repair
  flip restores one tour exactly when its chord crosses \(x\); the four
  intervals are reassembled without reversing orientation.}
  \label{fig:crossing-move}
\end{figure}
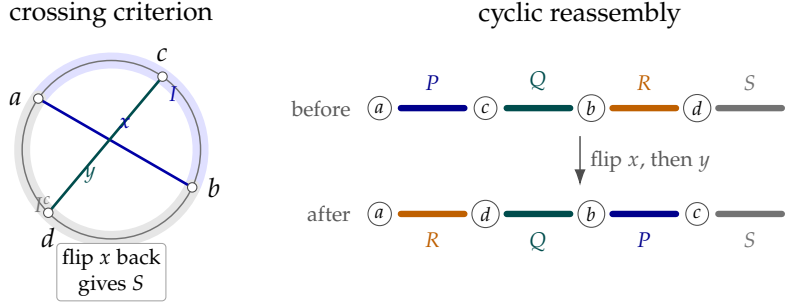

\subsection{Chunk Structure}

Fix a chunk parameter \(B\), which will be set to \(\lceil\sqrt M\rceil\) at
the end.  Maintain the cyclic sequence in chunks of size between \(B/2\) and
\(2B\), except during a local rebuild.  Let \(q=O(M/B)\) be the number of
chunks.  Give each live chunk a stable identifier from a pool of
\(Q=O(M/B)\) identifiers, with constant-factor slack for the temporary chunks
created during one update.  The aggregate vectors used by the data structure
are indexed by this identifier pool, so splitting or merging chunks does not
require resizing them.

When we rebuild a local window of chunks, we first give the new chunks fresh
identifiers.  Then, for every vertex with at least one occurrence in the
rebuilt window, we record its old unordered pair of chunk identifiers and its
new unordered pair, delete the old pair-list entry and endpoint contributions,
and insert the new ones.  This vertex-by-vertex update is important when both
occurrences of a vertex lie in the rebuilt window: the old pair is removed once
and the new pair is inserted once, so coincident identifiers cause no ambiguity.
Only after all such updates have been performed do we retire the old chunk
identifiers.  At that point no live vertex has an occurrence in a retired
identifier, and every aggregate contribution involving that identifier has been
deleted, because aggregate entries are maintained as a sum of these per-vertex
endpoint contributions.  We then clear the \(O(Q)\) pair-list headers incident
to each retired identifier, or equivalently bump a generation counter for those
headers, before putting the identifier back on the free list.  No global scan of
the treap vectors is needed; the zero aggregate coordinates are a consequence of
the deleted endpoint contributions.

For every occurrence, store its chunk, offset, and owner vertex.  For every
vertex, store handles to its two occurrences.  For every unordered pair of
chunks \(\{C,D\}\), allowing \(C=D\), maintain a dynamic array
\[
  L[C,D]
  =
  \{v:\text{ one occurrence of }v\text{ lies in }C
  \text{ and the other lies in }D\}.
\]
When \(C\ne D\), this list is stored only once but is addressable as both
\(L[C,D]\) and \(L[D,C]\); when \(C=D\), it contains the vertices whose two
occurrences are both in \(C\).  Each vertex stores its index in its current pair
list, so moving it between lists takes \(O(1)\) time by swap-delete.  The total
size of all lists is \(O(M)\), and there are \(O(Q^2)=O(M^2/B^2)\) headers.
The first flip of each step is sampled from a static array of the degree-two
vertices, since the vertex set does not change during the walk.
Thus the only nontrivial sampling problem is the repair flip: sample uniformly
from the queried chord and the chords crossing it.

The cyclic order of chunks is represented by any balanced binary sequence tree;
for concreteness one may use a treap \cite{SeidelAragon1996}.  Each internal
node \(U\) stores a vector indexed by chunk identifiers.  The invariant used by
the crossing query is the following: if \(D\) is not a chunk of the subtree
\(U\), then
\[
  \mathbf V_U[D]
  =
  \#\{v:\text{ one occurrence of }v\text{ is in a chunk of }U
  \text{ and the other occurrence is in }D\}.
\]
It is helpful to view the vector as storing endpoint contributions.  A vertex
with occurrences in chunks \(C\) and \(D\) contributes one unit to coordinate
\(D\) on every ancestor of \(C\), and one unit to coordinate \(C\) on every
ancestor of \(D\); when \(C=D\), these two contributions are combined on the
same ancestor path.  For a query we only read coordinates \(D\) outside the
subtree \(U\), where this contribution rule is exactly the displayed invariant.
Coordinates \(D\) lying inside \(U\) are maintained by the same bottom-up
summation rule so that later splits and concatenations can expose them as
outside coordinates of some new node.
Splitting or concatenating the treap recomputes \(O(\log q)\) internal nodes,
each in \(O(Q)\) time, for \(O(Q\log q)\) time.
For an affected vertex whose old chunk pair is \(\{C,D\}\) and whose new chunk
pair is \(\{C',D'\}\), we update the pair lists and adjust exactly these
endpoint contributions: delete one unit in coordinate \(D\) along the ancestor
path of \(C\) and one unit in coordinate \(C\) along the ancestor path of \(D\),
then insert the analogous contributions for \(C',D'\).  If some identifiers
coincide, the corresponding increments are combined.  This costs \(O(\log q)\)
per affected vertex.

\subsection{Sampling Crossings}

Given a query chord \(x\), let \(I=(a_x,b_x)\) be one of the two open intervals.
Decompose \(I\) into at most two partial boundary chunks and a set \(F\) of full
chunks.  The set \(F\) may be empty, and the two boundary chunks may coincide.
Let \(\mathcal B\) be the set of boundary chunks; see \cref{fig:chunk-query}.
Membership in \(I\) always means membership in this open interval, so the two
endpoints of the queried chord \(x\) are not counted as lying in \(I\).
The choice of which of the two intervals to call \(I\) does not affect the
crossing set: replacing \(I\) by its complement simply swaps the roles of
``inside'' and ``outside.''

Every crossing vertex other than \(x\) has exactly one endpoint in \(I\) and
one endpoint outside \(I\).  We split these vertices into two disjoint classes,
according to whether one of the two endpoints lies in a boundary chunk of the
interval decomposition.  This is only a data-structural split; the flip--repair
step must sample uniformly from the union of the two classes and \(x\) itself.
The classes are exhaustive because every endpoint of a crossing vertex lies
either in a boundary chunk or in a full inside/outside chunk, and they are
disjoint by definition.
Thus class A consists exactly of crossing vertices whose inside endpoint lies
in a full chunk of \(F\) and whose outside endpoint lies in a full chunk outside
\(F\cup\mathcal B\); class B is everything involving at least one boundary
endpoint.

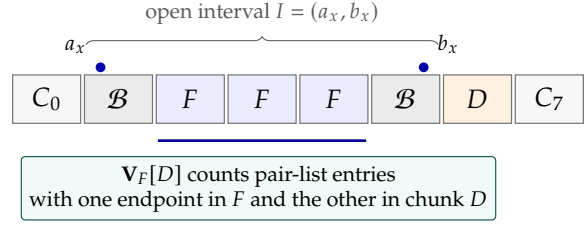
\begin{figure}[t]
  \centering
  \begin{tikzpicture}[
    x=.95cm,
    y=1cm,
    >=Latex,
    every node/.style={font=\small},
    chunk/.style={draw=black!55, line width=.45pt, minimum width=.9cm, minimum height=.64cm, inner sep=0pt},
    boundary/.style={chunk, fill=black!8},
    full/.style={chunk, fill=blue!8},
    outside/.style={chunk, fill=black!3},
    target/.style={chunk, fill=orange!12},
    arr/.style={-{Latex[length=2mm,width=1.4mm]}, line width=.55pt, black!65}
  ]
    \foreach \x/\style/\lab in {
      0/outside/\(C_0\),
      1/boundary/\(\mathcal B\),
      2/full/\(F\),
      3/full/\(F\),
      4/full/\(F\),
      5/boundary/\(\mathcal B\),
      6/target/\(D\),
      7/outside/\(C_7\)
    } {
      \node[\style] (c\x) at (\x,0) {\lab};
    }

    \draw[decorate, decoration={brace, amplitude=3pt}, black!65]
      (.55,.72) -- (5.45,.72)
      node[midway, above=4pt, font=\scriptsize] {open interval \(I=(a_x,b_x)\)};
    \node[circle, fill=blue!65!black, inner sep=1.3pt, label=above left:{\scriptsize \(a_x\)}] at (.75,.42) {};
    \node[circle, fill=blue!65!black, inner sep=1.3pt, label=above right:{\scriptsize \(b_x\)}] at (5.25,.42) {};

    \draw[blue!60!black, line width=1pt] (1.55,-.55) -- (4.45,-.55)
      node[midway, below=2pt, font=\scriptsize] {full chunks \(F\)};
    \draw[arr] (4.55,-.95) -- node[midway, below=2pt, font=\scriptsize]
      {coordinate \(D\)} (5.75,-.95);
    \node[draw=teal!55!black, fill=teal!5, rounded corners=1.5pt,
      align=center, font=\scriptsize, inner sep=3pt] at (2.95,-1.18)
      {\(\mathbf V_F[D]\) counts pair-list entries\\with one endpoint in \(F\)
      and the other in chunk \(D\)};
  \end{tikzpicture}
  \caption{A crossing query in the chunk data structure.  Boundary chunks are
  inspected explicitly.  For full chunks \(F\), the treap aggregate
  \(\mathbf V_F[D]\) counts how many chord endpoints in \(F\) are paired with
  endpoints in each outside chunk \(D\).}
  \label{fig:chunk-query}
\end{figure}

\paragraph{Class A.}
No endpoint lies in a boundary chunk.  Then one endpoint lies in a full inside
chunk \(C\in F\), and the other lies in a full outside chunk
\(D\notin F\cup\mathcal B\).  Cover the contiguous set of full chunks \(F\) by the
usual \(O(\log q)\) treap nodes, and let \(\mathbf V_F\) be the sum of their
stored vectors.  For every outside chunk \(D\notin F\cup\mathcal B\), the invariant
above gives
\[
  \mathbf V_F[D]
  =
  \#\{v:\text{ one endpoint of }v\text{ lies in }F,\text{ the other in }D\}.
\]
Because \(D\) is outside \(F\), every such vertex appears in exactly one pair
list \(L[C,D]\) with \(C\in F\), and conversely every vertex in one of these
lists has exactly one endpoint in the open interval \(I\).  Thus
\[
  \mathbf V_F[D]=\sum_{C\in F}\card{L[C,D]}.
\]
Let
\[
  W_A=\sum_{D\notin F\cup\mathcal B}\mathbf V_F[D].
\]
If \(W_A>0\), sample class A as follows: choose \(D\) with probability
\(\mathbf V_F[D]/W_A\), choose an inside chunk \(C\in F\) with probability
\(\card{L[C,D]}/\mathbf V_F[D]\), and then choose a uniformly random vertex
from \(L[C,D]\).  The second choice can be implemented by scanning the
\(O(q)\) chunk headers in \(F\), using the stored sizes of the lists
\(L[C,D]\); this scan is within the per-step budget.
For a fixed class-A vertex in the list \(L[C,D]\), the probability of selecting
it is
\[
  \frac{\mathbf V_F[D]}{W_A}\cdot
  \frac{\card{L[C,D]}}{\mathbf V_F[D]}\cdot
  \frac1{\card{L[C,D]}}
  =
  \frac1{W_A},
\]
so this procedure is uniform over class A.

\paragraph{Class B.}
At least one endpoint lies in a boundary chunk.  Enumerate all occurrences in
the boundary chunks, at most \(O(B)\) occurrences.  For each owner vertex \(v\),
test whether exactly one of \(a_v,b_v\) lies in \(I\).  The test uses the chunk
ranks in the top-level sequence tree and the offsets inside the two boundary
chunks; computing the needed ranks costs \(O(\log q)\) per tested vertex, or
\(O(B\log q)\) total for the boundary scan.  Deduplicate using timestamps,
yielding an explicit list \(Q_B\) of size \(W_B\).
The test removes noncrossing boundary vertices and also removes \(x\), because
the endpoints of \(x\) are the open-interval endpoints and are not counted as
lying in \(I\).  Sampling class B means choosing a uniformly random vertex from
this explicit list \(Q_B\).

Let \(W=W_A+W_B\).  The valid repair flips are \(x\) itself and the \(W\)
crossing vertices.  Choose \(x\) with probability \(1/(W+1)\), which gives a
self-loop.  Equivalently, if \(W=0\) the repair flip is forced to be \(x\).
If the chain does not choose \(x\), then \(W>0\); choose class A with
probability \(W_A/W\) and class B with probability \(W_B/W\), and sample inside
the chosen class as above.  A class-A vertex has probability
\((W_A/W)\cdot(1/W_A)=1/W\), and a class-B vertex has probability
\((W_B/W)\cdot(1/W_B)=1/W\).  Thus, conditional on not choosing the self-loop,
the repair flip is uniform over all crossing vertices, and unconditionally it
is uniform over all feasible repair choices adjacent to the odd intermediate
state.  Each such choice is a distinct coordinate flip from the odd set, so
sampling a vertex is the same as sampling the corresponding feasible even
neighbor of the flip--repair transition.

\subsection{Tour Update}

If the sampled repair flip is \(x\), the chain stays at the current tour.
Otherwise, suppose the four involved occurrences have cyclic order
\[
  a\; P\; c\; Q\; b\; R\; d\; S,
\]
where \(a,b\) belong to \(x\) and \(c,d\) belong to \(y\).  Flipping \(x\) and
\(y\) changes the cyclic order to
\[
  a\; R\; d\; Q\; b\; P\; c\; S.
\]
This is a constant number of cuts and concatenations; no reversal is needed in
the directed successor-swap convention.  We first split chunks at the four cut
positions, perform the treap update, and then rebuild the boundary chunks.  Only
\(O(B)\) occurrences change chunk membership during this rebuild.
The pair lists depend only on the two chunk identifiers containing a vertex's
occurrences.  Therefore merely reordering whole chunks during the cyclic
reassembly does not change any pair-list membership; only occurrences lying in
newly split or rebuilt chunks need to be moved between pair lists and reflected
in the treap aggregates.

\begin{lemma}[Local chunk rebuilding]\label{lem:local-rebuild}
After one nontrivial flip--repair update, the chunk invariant can be restored by
changing chunk identifiers for \(O(B)\) occurrences and rebuilding \(O(1)\)
chunks.
\end{lemma}

\begin{proof}
Before the update, every chunk has size at most \(2B\).  The four cut positions
therefore lie in at most four old chunks of total size at most \(8B\).  Splitting
at the cut positions creates only pieces of these old chunks; every other chunk
is moved as an intact block during the cyclic reassembly and keeps its
identifier.

After the reassembly, all chunks that can violate the size lower bound are among
these \(O(1)\) split pieces.  For each maximal run of such boundary pieces, add
its immediate left and right neighboring chunks, unless the whole cyclic
sequence is already contained in the run.  Coalesce overlapping windows.  The
result is \(O(1)\) disjoint windows of total length \(O(B)\): they contain the
pieces of the four old cut chunks and at most two neighboring chunks of size at
most \(2B\) per run.

Repartition each window into consecutive chunks of sizes in \([B/2,2B]\), except
in the degenerate case \(M<B/2\), where there is only one chunk and all
operations are constant time.  One concrete rule is to cut chunks of size \(B\)
from left to right and, if the final remainder has size below \(B/2\), merge it
with the previous new chunk; the merged chunk has size below \(3B/2\).  If a
window had length below \(B/2\), the added neighbor makes its length at least
\(B/2\).  Since each rebuilt window has length \(O(B)\), it creates \(O(1)\)
chunks and changes identifiers for only \(O(B)\) occurrences in total.
\end{proof}

\begin{theorem}[Degree-two implementation]\label{thm:data-structure}
For a \(2\)-in/\(2\)-out directed Eulerian graph with \(M\) arcs, the
degree-two flip--repair chain can be implemented with
\[
  O\!\left(M+\frac{M^2}{B^2}\right)
\]
space and
\[
  O\!\left(B\log M+\frac{M}{B}\log M\right)
\]
amortized time per step.  With \(B=\lceil\sqrt M\rceil\), this is \(O(M)\)
space and \(O(\sqrt M\log M)\) time per step.
\end{theorem}

\begin{proof}
The pair lists \(L[C,D]\) contain each vertex once, so their total size is
\(O(M)\).  The chunk-identifier pool has size \(Q=O(M/B)\), so there are
\(O(Q^2)=O(M^2/B^2)\) list headers.  The treap stores one length-\(Q\) vector
at each internal node.  Since the top-level tree has \(O(q)=O(Q)\) nodes, these
vectors take \(O(Q^2)\) space.  The total space is
\(O(M+Q^2)=O(M+M^2/B^2)\), which is \(O(M)\) when \(B\) is chosen on the order
of \(\sqrt M\).
The initial chunks, pair lists, and aggregate vectors can be built within the
same asymptotic space and time by a single pass over the tour followed by a
bottom-up construction of the treap.

For a query chord \(x\), locate the two chunks containing \(a_x\) and \(b_x\).
The open interval \(I=(a_x,b_x)\) consists of two partial boundary chunks and a
contiguous interval of full chunks \(F\) in the top-level cyclic order.  The
full chunks in \(F\) are covered by \(O(\log q)\) treap nodes, without changing
the chunk decomposition.  If the cyclic interval crosses the chosen
linearization point of the cyclic order, we cover it as two ordinary intervals,
which preserves the same asymptotic bound.  Summing their stored vectors costs
\(O(Q\log q)\).
The class-A chunk pair is then sampled by a prefix-sum scan over the \(q\)
possible outside chunks and, after the outside chunk \(D\) is chosen, a scan
over the inside chunks \(C\in F\).  This additional scan costs \(O(q)\).  Class
B only inspects the two boundary chunks, hence \(O(B)\) occurrences, and the
interval-membership tests for these occurrences cost \(O(B\log q)\).  Thus
sampling a valid repair flip costs
\[
  O(Q\log q+q+B\log q)=O(Q\log q+B\log q),
\]
because \(q\le Q\).

The tour update performs a constant number of treap splits and concatenations.
Each affected treap node recomputes a length-\(Q\) vector, for
\(O(Q\log q)\) time.  By \cref{lem:local-rebuild}, the only occurrences whose
chunk identifiers can change lie in \(O(1)\) rebuilt windows of total size
\(O(B)\).  Thus one update affects \(O(B)\) occurrences and hence \(O(B)\)
vertices.  For each affected vertex we update its pair-list entry in \(O(1)\)
time and perform the old-pair deletion and new-pair insertion in the aggregate
vectors on \(O(\log q)\) ancestor nodes as described above.  Occurrences in
chunks that are only cut, concatenated, or moved as whole chunks keep the same
pair-list entry.  Finally, for the \(O(1)\) retired identifiers, clearing the
incident pair-list headers or bumping their generation counters costs \(O(Q)\)
time, already within the query and treap-update budget.
Thus rebuilding contributes \(O(B\log q+Q)\) time to the step that performs it.
The identifier-free-list convention from \cref{sec:data-structure} is safe
because all per-vertex endpoint contributions involving an old rebuilt chunk are
deleted before the old identifier is recycled.
Substituting \(q=O(M/B)\) and \(Q=O(M/B)\) proves the theorem.
\end{proof}

\begin{proof}[Proof of \cref{thm:degree-two-sampler}]
If \(H\) has one vertex, then it has two loop arcs.  Of the two local
transitions, exactly one orders these two arcs in a single directed cycle; the
algorithm outputs that transition system.
Assume from now on that \(H\) has at least two vertices.
Find an initial Eulerian tour \(C\) of \(H\) by Hierholzer's algorithm
\cite{Hierholzer1873}.  By \cref{thm:tour-pfaffian}, representing transition
systems by the set of vertices flipped relative to \(C\) identifies the uniform
measure on \(\ET(H)\) with a flat skew-determinantal measure.  Start the chain
at \(S=\varnothing\), the state corresponding to \(C\).  The matrix \(A_C\) is
only an analysis device: the implementation of a step uses the current tour and
the crossing criterion of \cref{lem:crossing}, not the entries of \(A_C\).  By
\cref{thm:flat-lsi}, the flip--repair chain reaches total variation distance
\(\eps\) from stationarity after
\[
  O(M\log M[\log M+\log(1/\eps)])
\]
steps.  Here the number of coordinates is \(n=\card{V(H)}\), and
after the one-vertex loop case has been removed we have \(M=2n\); the displayed
bound is the same, up to constants, if one simply uses \(n\le M\).  By
\cref{thm:data-structure}, each step costs \(O(\sqrt M\log M)\) amortized time
and the maintained data structure uses \(O(M)\) space.  The initial tour and the
initial chunk data structure are built in \(O(M)\) time, which is lower order
than the walk length.  Multiplying the number of steps by the per-step cost
gives the claimed running time.
\end{proof}

\section{A Standard Degree-Reduction Gadget}\label{sec:gadgets}

It remains to connect the degree-two sampler to a graph with arbitrary
Eulerian degrees.  The reduction is local.  If
\(d=\deg^+(v)=\deg^-(v)\), then a transition system chooses a permutation from
the \(d\) incoming arcs of \(v\) to the \(d\) outgoing arcs.  We replace this
local choice by a degree-two switching network.
Vertices with \(d=1\) have a forced local transition; we contract them before
the reduction and expand them again after sampling.
This is a supporting reduction, not the main mixing argument.  It is in the
same spirit as earlier local gadget reductions for Eulerian tours and A-trails
\cite{GeStefankovic2012} and the degree-reduction step used for Eulerian tours
in parallel discrete sampling via continuous walks
\cite{AnariHuangLiuVuongXuYu2023}.
We first check whether every vertex has degree one.  In that all-forced case,
strong connectivity makes the graph a directed cycle and there is nothing to
randomize.  Otherwise, at least one vertex has degree at least two, and
suppressing a degree-one vertex does not change the degrees of any other vertex.
Thus a non-forced vertex remains throughout the suppression process, and after
all forced vertices have been suppressed the remaining graph has minimum degree
at least two.

The gadget reduction has two separate obligations.  First, contracting each
gadget back to the original vertex must preserve the number of cycles in the
transition system.  Second, for each original local permutation, the number of
expanded local settings projecting to it must be almost independent of the
permutation.  The next lemmas prove these two facts.

\begin{lemma}[Suppressing a forced vertex]\label{lem:suppress-degree-one}
Let \(v\) be a degree-one vertex in a directed Eulerian multigraph.  Let \(e\)
be the unique arc entering \(v\), and let \(f\) be the unique arc leaving
\(v\).  Assume \(e\ne f\); in a strongly connected graph this excludes only the
single-vertex single-loop case.  Form \(G/v\) by deleting \(e\) and \(f\) and
inserting one new arc \(g\) from the tail of \(e\) to the head of \(f\).  Then
transition systems of \(G\) are in bijection with transition systems of
\(G/v\), and the bijection preserves the number of cycles.  In particular, it
preserves the property of being one Eulerian tour.
\end{lemma}

\begin{proof}
At \(v\), every transition system must use the forced local successor
\(e\mapsto f\).  Given a transition system of \(G\), modify only the two local
bijections adjacent to the contracted path: at the tail of \(e\), replace the
outgoing incidence of \(e\) in the codomain by the outgoing incidence of the new
arc \(g\); at the head of \(f\), replace the incoming incidence of \(f\) in the
domain by the incoming incidence of \(g\).  All other local choices are
unchanged.  This gives a transition system of \(G/v\), and the operation is
reversible by expanding \(g\) back into the forced pair \(e,f\).

In the arc permutation of \(G\), the arcs \(e\) and \(f\) are consecutive in
every cycle because \(\pi(e)=f\).  Contracting this consecutive pair to the
single arc \(g\) does not split or merge cycles; it only shortens one cycle by
one arc.  Thus the number of cycles is preserved.
\end{proof}

A \(D\)-wire switching network has \(D\) input ports, \(D\) output ports, and
internal vertices called switches.  Ports are attachment points, not vertices of
the expanded graph.  The arcs of the network are directed wire segments running
from input ports through switches and then to output ports.  Each switch has two
incoming wire segments and two outgoing wire segments, so it is a
\(2\)-in/\(2\)-out vertex.  A switch setting is one of the two local successor
pairings at that vertex: straight or crossed.  Once every switch is set, every
input port is routed along a unique directed path to an output port, and these
paths define a permutation of \([D]\).

The networks we use are sparse random lazy-transposition networks with a short
guard layer whose locations are deterministic.  First insert switches between
wires
\[
  (1,2),(2,3),\ldots,(D-1,D).
\]
This touches every wire at least once.  Then choose
\[
  R=2\left\lceil C\,D\bigl(\log D+\log(1/(\eta\delta))\bigr)\right\rceil
\]
unordered pairs \(a_tb_t\) from \(\binom{[D]}2\), independently and uniformly.
For each chosen pair, insert one switch between the current wire \(a_t\) and
the current wire \(b_t\).  Wires that are not involved in a switch continue to
their next switch or output port as a single directed arc segment.  Thus the
network has \(D-1+R\) internal vertices and \(O(D+R)\) arc segments.  A setting
of the switches is exactly a product of lazy transpositions in the prescribed
order.  The guard layer is included only to make the graph replacement
well-formed after port suppression; it does not hurt the pointwise mixing
bound, because convolving an already pointwise-close distribution with any
other distribution preserves pointwise closeness to uniform.

We use the following upper-bound-sum form of the Diaconis--Shahshahani analysis
of random transpositions \cite{DiaconisShahshahani1981}.  Let
\(\widehat{S_D}\) be the set of irreducible representations of \(S_D\).  For
\(\lambda\in\widehat{S_D}\), let \(d_\lambda\) be its dimension, \(\chi_\lambda\)
its character, and
\[
  \alpha_\lambda
  =
  \frac12\left(1+\frac{\chi_\lambda(\tau)}{d_\lambda}\right),
\]
where \(\tau\) is any transposition.  Thus \(\alpha_\lambda\) is the scalar by
which the lazy random-transposition kernel acts on the \(\lambda\)th
isotypic component.  In particular, \(\alpha_\lambda\in[0,1]\), since these
scalars are eigenvalues of a lazy reversible Markov kernel.

\begin{lemma}[Diaconis--Shahshahani upper-bound sum]\label{lem:ds-upper}
There is an absolute constant \(C_{\rm DS}\) such that for every \(D\ge2\) and
every \(0<\xi<1\), if
\[
  k\ge C_{\rm DS}D\bigl(\log D+\log(1/\xi)\bigr),
\]
then
\[
  \sum_{\lambda\in\widehat{S_D},\,\lambda\ne(D)}
  d_\lambda^2\alpha_\lambda^k
  \le
  \xi .
\]
\end{lemma}

\begin{proof}
This is the standard Fourier upper-bound lemma together with the
Diaconis--Shahshahani character estimates.  For completeness, we spell out the
normalization.  If \(Q\) is the lazy random-transposition kernel on \(S_D\), then
\(Q\) is central and Schur's lemma gives
\[
  \widehat Q(\lambda)=\alpha_\lambda I_{d_\lambda}
\]
on the \(\lambda\)th irreducible representation.  Fourier inversion therefore
gives, for the \(k\)-step law \(Q^k(e,\cdot)\),
\[
  \sup_{\sigma\in S_D}\abs{D!\,Q^k(e,\sigma)-1}
  \le
  \sum_{\lambda\ne(D)} d_\lambda^2 \alpha_\lambda^k .
\]
In their notation, this is the upper-bound sum in their equation (3.1), together
with the character-ratio estimates split into Zones I--III in their Section 3.
Those estimates show that the sum is at most \(\xi\) after
\(C D(\log D+\log(1/\xi))\) transpositions, for a universal constant \(C\).
They use only bounds on the transposition character ratio
\(\beta_\lambda=\chi_\lambda(\tau)/d_\lambda\).  Passing to the lazy walk
replaces \(\beta_\lambda\) by
\(\alpha_\lambda=(1+\beta_\lambda)/2\), so
\(1-\alpha_\lambda=(1-\beta_\lambda)/2\).  The same zone estimates therefore
apply with a different absolute constant.
\end{proof}

Equivalently, the \(k\)-step lazy random-transposition law \(p_k\) satisfies
\[
  \sup_{\sigma\in S_D}\abs{D!\,p_k(\sigma)-1}\le \xi
\]
under the same hypothesis, because Fourier inversion bounds the left-hand side
by the displayed character sum.  We use this character-sum form because it
passes cleanly through the second-moment calculation for a fixed random
switching network.

\begin{proof}[Proof of \cref{thm:pointwise-switching}]
Use the random network described at the start of the section with
\[
  C\ge C_{\rm DS}
\]
large enough for \cref{lem:ds-upper} with
\(\xi=\eta\delta/2\), and write \(R=2k\).  We first analyze only the random
suffix.  For a fixed realized suffix
\(\cR=(\tau_1,\ldots,\tau_R)\), let \(\mu_{\cR}\) be the distribution on \(S_D\)
obtained from uniform independent settings of these \(R\) switches.

We prove that, with probability at least \(1-\delta\) over the random
transpositions \(\tau_t\),
\[
  \max_{\sigma\in S_D}\abs{D!\,\mu_{\cR}(\sigma)-1}\le \eta/2.
  \tag{*}\label{eq:quenched-switching}
\]
Fix a unitary irreducible representation
\(\rho_\lambda:S_D\to U(d_\lambda)\).  The Fourier transform of
\(\mu_{\cR}\) at \(\lambda\) is
\[
  \widehat\mu_{\cR}(\lambda)
  =
  P_{\lambda,R}P_{\lambda,R-1}\cdots P_{\lambda,1},
  \qquad
  P_{\lambda,t}=\frac{I+\rho_\lambda(\tau_t)}2 .
\]
Each \(P_{\lambda,t}\) is an orthogonal projection, since \(\tau_t^2=e\).  Also,
because a uniformly random transposition is a conjugacy-class average, Schur's
lemma gives
\[
  \E_{\tau_t}{P_{\lambda,t}}=\alpha_\lambda I .
\]
We claim that
\begin{equation}\label{eq:random-product-hs}
  \E_{\cR}{\norm{\widehat\mu_{\cR}(\lambda)}_{\mathrm{HS}}^2}
  =
  d_\lambda\alpha_\lambda^R .
\end{equation}
Indeed, if \(M_t=P_{\lambda,t}\cdots P_{\lambda,1}\), then, conditioning on
\(M_{t-1}\),
\[
  \E_{\tau_t}{\norm{M_t}_{\mathrm{HS}}^2\mid M_{t-1}}
  =
  \E_{\tau_t}{\operatorname{Tr}(P_{\lambda,t}M_{t-1}M_{t-1}^*P_{\lambda,t})}
  =
  \alpha_\lambda\norm{M_{t-1}}_{\mathrm{HS}}^2,
\]
where we used \(P_{\lambda,t}^2=P_{\lambda,t}\) and cyclicity of trace.
Starting from \(M_0=I\) proves \cref{eq:random-product-hs}.

Fourier inversion on \(S_D\) gives, for every fixed suffix and every
\(\sigma\in S_D\),
\[
  D!\,\mu_{\cR}(\sigma)-1
  =
  \sum_{\lambda\ne(D)}
  d_\lambda\operatorname{Tr}\!\left(
    \rho_\lambda(\sigma^{-1})\widehat\mu_{\cR}(\lambda)
  \right).
\]
Since \(\rho_\lambda(\sigma^{-1})\) is unitary,
\[
  \abs*{\operatorname{Tr}\!\left(
    \rho_\lambda(\sigma^{-1})\widehat\mu_{\cR}(\lambda)
  \right)}
  \le
  \sqrt{d_\lambda}\,
  \norm{\widehat\mu_{\cR}(\lambda)}_{\mathrm{HS}}.
\]
Thus the random variable
\[
  B(\cR)
  =
  \sum_{\lambda\ne(D)}
  d_\lambda^{3/2}
  \norm{\widehat\mu_{\cR}(\lambda)}_{\mathrm{HS}}
\]
dominates \(\max_\sigma\abs{D!\mu_{\cR}(\sigma)-1}\).  By Jensen's inequality
and \cref{eq:random-product-hs},
\[
  \E_{\cR}{B(\cR)}
  \le
  \sum_{\lambda\ne(D)}
  d_\lambda^{3/2}\sqrt{d_\lambda\alpha_\lambda^R}
  =
  \sum_{\lambda\ne(D)}
  d_\lambda^2\alpha_\lambda^k
  \le
  \frac{\eta\delta}{2},
\]
where the last inequality is \cref{lem:ds-upper}.  Markov's inequality gives
\(\P{B(\cR)>\eta/2}\le\delta\), proving \cref{eq:quenched-switching}.

On this good event,
\[
  1-\eta/2
  \le
  D!\,\mu_{\cR}(\sigma)
  \le
  1+\eta/2
  \qquad\forall\sigma\in S_D.
\]
Since \(1-\eta/2\ge e^{-\eta}\) and \(1+\eta/2\le e^\eta\) for
\(0<\eta<1\), the random suffix alone has the required pointwise
multiplicative bound.

Let \(\nu\) be the distribution induced by the deterministic guard layer under
uniform independent switch settings.  The full network law is either
\(\nu*\mu_{\cR}\) or \(\mu_{\cR}*\nu\), depending on the convention for
composition.  In either case, pointwise bounds for \(\mu_{\cR}\) imply the same
pointwise bounds for the convolution, because
\[
  \sum_{\pi\in S_D}\nu(\pi)\frac{e^{-\eta}}{D!}
  \le
  (\nu*\mu_{\cR})(\sigma)
  \le
  \sum_{\pi\in S_D}\nu(\pi)\frac{e^\eta}{D!}.
\]
The proof for \(\mu_{\cR}*\nu\) is identical.  Hence the full guarded network
satisfies the theorem.
\end{proof}

We now use this network as a graph replacement.  For a degree-\(d\) vertex
\(v\), apply \cref{thm:pointwise-switching} with \(D=d\).  Delete \(v\).
For each incoming arc \(i_j\), keep its tail and make its head the \(j\)th input
port.  For each outgoing arc \(o_j\), make the \(j\)th output port its tail and
keep its head.  The ports only splice original arcs to network wire segments.
After all splices are made, suppress the ports:
each maximal directed path whose internal points are ports and whose endpoints
are switches is replaced by one arc.  The guard layer ensures that every wire
has at least one switch, so port suppression does not create a closed directed
component with no switch vertex.  After suppression, the only vertices left in
the expanded graph are switches, and every such vertex has indegree and
outdegree two.  A realized network can still fail the pointwise guarantee of
\cref{thm:pointwise-switching}; that failure is not detected by the algorithm
and is accounted for only by the bad construction event in the analysis.
\Cref{fig:switching-gadget} shows this local replacement.

\begin{figure}[t]
  \centering
  \begin{tikzpicture}[
    scale=1.06,
    transform shape,
    x=1cm,
    y=.76cm,
    >=Latex,
    every node/.style={font=\small},
    wire/.style={line width=.7pt, blue!60!black},
    port/.style={circle, draw=black!58, fill=white, inner sep=1.05pt},
    switchbox/.style={draw=black!60, fill=white, rounded corners=1pt, line width=.45pt},
    arrow/.style={-{Latex[length=2mm,width=1.4mm]}, line width=.6pt},
    label/.style={font=\scriptsize, text=black!70}
  ]
    \draw[black!30, rounded corners=2pt, fill=black!2] (.45,.18) rectangle (7.15,4.28);
    \node[font=\scriptsize, fill=black!2, inner sep=2pt] at (3.8,4.28)
      {guarded random lazy-transposition network};

    \foreach \y/\lab/\outlab in {
      3.56/\(i_1\)/\(o_1\),
      2.68/\(i_2\)/\(o_2\),
      1.80/\(\cdots\)/\(\cdots\),
      .92/\(i_d\)/\(o_d\)
    } {
      \draw[arrow, blue!60!black] (-.25,\y) -- (.75,\y);
      \draw[wire] (.75,\y) -- (6.85,\y);
      \draw[arrow, blue!60!black] (6.85,\y) -- (7.95,\y);
      \node[anchor=east, label] at (-.32,\y) {\lab};
      \node[anchor=west, label, text=blue!60!black] at (8.02,\y) {\outlab};
      \node[port] at (.75,\y) {};
      \node[port] at (6.85,\y) {};
    }

    \foreach \x/\ya/\yb in {
      1.55/3.56/2.68,
      2.35/.92/2.68,
      3.15/3.56/.92,
      4.05/2.68/1.80,
      4.95/3.56/2.68,
      5.85/1.80/.92
    } {
      \pgfmathsetmacro{\ym}{(\ya+\yb)/2}
      \pgfmathsetmacro{\ytop}{max(\ya,\yb)+.26}
      \pgfmathsetmacro{\ybot}{min(\ya,\yb)-.26}
      \path[switchbox] (\x-.24,\ybot) rectangle (\x+.24,\ytop);
      \draw[black!45, line width=.45pt] (\x-.13,\ya) -- (\x+.13,\yb);
      \draw[black!45, line width=.45pt] (\x-.13,\yb) -- (\x+.13,\ya);
    }

    \node[label, align=center] at (3.8,.38)
      {each switch is independently straight or crossed};
  \end{tikzpicture}
  \caption{A degree-\(d\) vertex is replaced by a \(d\)-wire sparse switching
  network.  A short guard layer touches every wire, and a random construction
  chooses the remaining pairs of wires that meet at switches.  Uniform switch
  settings then induce a pointwise almost-uniform permutation of the \(d\) local
  successors.}
  \label{fig:switching-gadget}
\end{figure}
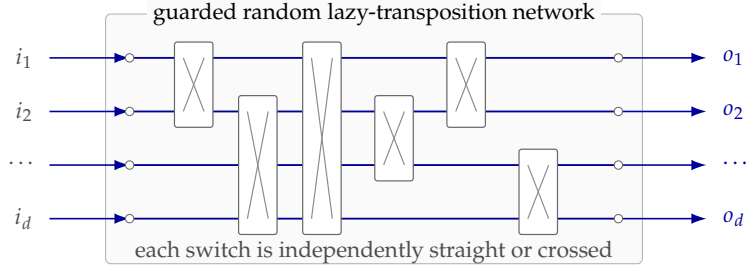

\begin{lemma}[Projection preserves cycle structure]\label{lem:gadget-cycles}
Fix an original transition system \(T\) of \(G\).  For each vertex \(v\), choose
a network setting whose induced permutation is the local permutation of \(T\) at
\(v\).  Let \(T^\star\) be the expanded transition system obtained by these
gadget choices.  Then \(T^\star\) and \(T\) have the same number of directed
cycles.  Conversely, every expanded transition system that is an Eulerian tour
projects to an Eulerian tour of \(G\).
\end{lemma}

\begin{proof}
Inside one gadget, start at an input wire and follow the switch setting until
an output wire appears.  Contract this whole directed trajectory to the
corresponding original local transition.  Doing this independently in every
gadget changes only the subdivision of arcs in the permutation: it replaces each
maximal path through a gadget by one local successor choice at the original
vertex.  Subdividing or contracting arcs of a directed cycle cannot split or
merge cycles, so \(T^\star\) and \(T\) have the same number of cycles.

Conversely, every expanded transition system induces a well-defined local
permutation at each original vertex by following the \(d\) input wires of that
gadget to the \(d\) output wires.  If the expanded system has one cycle, then
contracting the gadget trajectories gives one cycle in the original graph.
\end{proof}

\begin{lemma}[Pointwise local lift counts]\label{lem:lift-count}
Assume a \(d\)-wire gadget has \(r\) switches and satisfies the pointwise bound
of \cref{thm:pointwise-switching} with error \(\eta\).  If \(\ell(\pi)\) is the
number of switch settings inducing \(\pi\in S_d\), then
\[
  e^{-\eta}\frac{2^r}{d!}
  \le
  \ell(\pi)
  \le
  e^\eta\frac{2^r}{d!}.
\]
\end{lemma}

\begin{proof}
Under uniform switch settings, the probability of inducing \(\pi\) is
\(\ell(\pi)/2^r\).  The pointwise bound says that this probability lies between
\(e^{-\eta}/d!\) and \(e^\eta/d!\).  Multiplying by \(2^r\) gives the claim.
\end{proof}

We can now define the projection used by the sampler.  Given an expanded
Eulerian tour of \(G^\star\), read off the network permutation realized inside
each gadget.  This is the local permutation at the corresponding vertex of the
suppressed graph.  Finally, undo the degree-one suppressions from
\cref{lem:suppress-degree-one}.  The result is a transition system of the
original graph \(G\), and by the cycle-preservation lemma it is again one
Eulerian tour.  We need two estimates from this construction: projection
probabilities must be nearly uniform over original tours, and the expanded
graph must remain nearly linear in size.

\begin{theorem}[Faithful degree-two expansion]\label{thm:faithful-expansion}
Let \(G\) be a strongly connected directed Eulerian multigraph.  Let
\(G^\star\) be the graph obtained as follows.  If every vertex has degree one,
do not add gadgets.  Otherwise, first repeatedly suppress degree-one vertices as
in \cref{lem:suppress-degree-one}; then replace every remaining vertex \(v\),
which now has degree \(d_v\ge2\), by a \(d_v\)-wire gadget satisfying the
pointwise bound of \cref{thm:pointwise-switching} with error \(\eta_v\),
constructed using failure parameter \(\delta_v\).
Let \(V_{\ge2}\) be the vertices that receive gadgets, and let
\(\eta_{\mathrm{tot}}=\sum_{v\in V_{\ge2}}\eta_v\).  If
\(\ET(G^\star)\ne\varnothing\), \(T^\star\) is uniform on \(\ET(G^\star)\), and
\(T\) is its projection to \(G\), then
\[
  \frac{e^{-2\eta_{\mathrm{tot}}}}{\card{\ET(G)}}
  \le
  \P{T=T_0}
  \le
  \frac{e^{2\eta_{\mathrm{tot}}}}{\card{\ET(G)}}
  \qquad\forall T_0\in\ET(G),
\]
and consequently
\[
  \norm*{
    \mathcal L(T)-\operatorname{Unif}(\ET(G))
  }_{\TV}
  \le
  e^{2\eta_{\mathrm{tot}}}-1.
\]
If \(V_{\ge2}\ne\varnothing\), the expanded graph has
\[
  M=O(m\,\polylog(m/(\eta_{\min}\delta_{\min}))),
  \qquad
  \eta_{\min}=\min_{v\in V_{\ge2}}\eta_v,
  \quad
  \delta_{\min}=\min_{v\in V_{\ge2}}\delta_v,
\]
arcs, where \(\delta_v\) is the failure probability used when constructing the
random gadget at \(v\).  If \(V_{\ge2}=\varnothing\), then all local
transitions are forced and no gadget is needed.
\end{theorem}

\begin{proof}
If every vertex of \(G\) has degree one, then strong connectivity implies that
\(G\) is a directed cycle.  There is a single transition system, it is an
Eulerian tour, and no gadget is needed.  The theorem is immediate in this case.

Let \(\bar G\) be the graph obtained after repeatedly suppressing degree-one
vertices until none remain.  Since we are not in the all-forced case, at least
one vertex has degree at least two; suppressing a degree-one vertex preserves
the degrees of all other vertices, so each suppression is covered by
\cref{lem:suppress-degree-one}.  Thus transition systems of \(G\) and
\(\bar G\) are in a cycle-preserving bijection.  The suppressions contribute no
multiplicative factor to the lift count, and each suppression replaces two arcs
by one, so the suppressed graph has at most the original \(m\) arcs.  It is
therefore enough to do the counting and size bound on \(\bar G\) and then
re-expand the forced paths at the end.  For readability, write \(G\) for
\(\bar G\) for the rest of the proof.

For a transition system \(T\) of this suppressed graph, let \(\pi_v(T)\) be its
local permutation at \(v\).  For \(v\in V_{\ge2}\), let \(\ell_v(\pi)\) be the
number of local switch settings in the gadget of \(v\) whose network
permutation is \(\pi\).
Equivalently, \(\ell_v(\pi)\) is \(2^{r_v}\) times the probability that the
\(r_v\) switches of the network, with uniform independent settings, induce
\(\pi\).  By
\cref{lem:lift-count}, the pointwise lower bound is positive for every
\(\pi\in S_{d_v}\), so every original local permutation is realized by at
least one switch setting.  Because the gadgets are vertex
disjoint after the original arcs are attached, and because there are no
additional local choices outside gadgets in the suppressed graph, the number of
expanded transition systems projecting to \(T\) factors as
\[
  L(T)=\prod_{v\in V_{\ge2}} \ell_v(\pi_v(T)),
\]
where
\[
  e^{-\eta_v}A_v\le \ell_v(\pi)\le e^{\eta_v}A_v
\]
for the constant
\[
  A_v=\frac{2^{r_v}}{d_v!},
\]
which is independent of \(\pi\).  Hence
\[
  e^{-\eta_{\mathrm{tot}}}\prod_{v\in V_{\ge2}} A_v
  \le
  L(T)
  \le
  e^{\eta_{\mathrm{tot}}}\prod_{v\in V_{\ge2}} A_v.
\]
The product \(L(T)\) counts expanded transition systems before conditioning on
the one-cycle event.  By \cref{lem:gadget-cycles}, the one-cycle condition is
preserved by projection.  Thus, if \(T\notin\ET(G)\), none of its expanded lifts
is an Eulerian tour; if \(T\in\ET(G)\), all of its expanded lifts are Eulerian
tours.  The converse direction of \cref{lem:gadget-cycles} rules out any
additional expanded Eulerian tours projecting outside \(\ET(G)\).
Therefore conditioning on the expanded transition system being one cycle simply
restricts the above lift-count comparison to \(T\in\ET(G)\).  If
\[
  C_0=\prod_{v\in V_{\ge2}} A_v,
\]
then every \(T\in\ET(G)\) has between \(e^{-\eta_{\mathrm{tot}}}C_0\) and
\(e^{\eta_{\mathrm{tot}}}C_0\) expanded Eulerian-tour lifts.  If
\(N=\card{\ET(G)}\), the total number of expanded tours is between
\(N e^{-\eta_{\mathrm{tot}}}C_0\) and \(N e^{\eta_{\mathrm{tot}}}C_0\).  Hence,
after normalization, each projected probability lies between
\(e^{-2\eta_{\mathrm{tot}}}/N\) and \(e^{2\eta_{\mathrm{tot}}}/N\).  Therefore,
if \(p\) is the projected law
and \(u\) is uniform on \(\ET(G)\), then
\[
  \norm{p-u}_{\TV}
  =
  \sum_{T:p(T)>u(T)}(p(T)-u(T))
  \le
  (e^{2\eta_{\mathrm{tot}}}-1)\sum_{T:p(T)>u(T)}u(T)
  \le
  e^{2\eta_{\mathrm{tot}}}-1.
\]
For the size bound, let \(d_v\) be the degree of \(v\) in the suppressed graph
and use a \(d_v\)-wire random network with failure parameter \(\delta_v\).
\Cref{thm:pointwise-switching} gives
\(O(d_v\polylog(d_v/(\eta_v\delta_v)))\) switches and arcs.  Summing over
vertices gives
\[
  \sum_{v\in V_{\ge2}} d_v\polylog(d_v/(\eta_v\delta_v))
  =
  O\!\left(m\,\polylog(m/(\eta_{\min}\delta_{\min}))\right),
\]
using \(\sum_v d_v=\card{E(G)}\le m\) for the suppressed graph, where \(m\)
denotes the number of arcs in the original input graph.
\end{proof}

\begin{proof}[Proof of \cref{thm:main}]
If every vertex has degree one, strong connectivity implies that \(G\) is a
directed cycle, so there is a single Eulerian tour and the algorithm outputs it.
Assume from now on that at least one vertex has degree at least two.
Choose gadget errors and construction-failure probabilities, for instance
\[
  \eta_v=\frac{\eps}{16m},
  \qquad
  \delta_v=\frac{\eps}{16m}
\]
for every non-forced vertex.  There are at most \(m\) such vertices, because
each has positive outdegree and \(\sum_v\deg^+(v)=m\).  Hence
\(\sum_v\eta_v\le \eps/16\), and by the union bound the event that all random
gadgets satisfy their pointwise guarantees has probability at least
\(1-\eps/16\).  Call this event \(\mathsf{Good}\), and call its complement
\(\mathsf{Bad}\).

Construct \(G^\star\) using these random gadgets.  On \(\mathsf{Good}\),
the same identity \(\sum_v \deg^+(v)=m\) gives that the switching
networks have total size
\[
  M=\card{E(G^\star)}=O(m\polylog(m/\eps)).
\]
\Cref{thm:pointwise-switching} gives the same bound on construction time, so
building the expanded graph does not dominate the final
\(M^{3/2}\)-type sampling cost.
On \(\mathsf{Good}\), the expanded graph has an Eulerian tour.
Indeed, by Hierholzer's theorem \cite{Hierholzer1873}, \(G\) has an Eulerian
tour.  First suppress the forced degree-one vertices of \(G\) as in
\cref{lem:suppress-degree-one}; the image of this tour is an Eulerian tour of
the suppressed graph.  For
each remaining vertex, the local permutation of this suppressed tour is
realized by at least one switch setting by \cref{lem:lift-count}.  Then
\cref{lem:gadget-cycles} gives a one-cycle transition system in \(G^\star\).
Every vertex of \(G^\star\) is a switch and has indegree and outdegree two, and
the one-cycle transition system traverses every arc exactly once.  Therefore,
for any two arcs \(e,e'\) of \(G^\star\), the tour order contains a directed
path from \(e\) to \(e'\); taking endpoints of arcs shows that the underlying
directed graph of \(G^\star\) is strongly connected.
If the expanded graph construction does not produce a strongly connected
degree-two graph, the algorithm falls back to the Eulerian tour of \(G\) found
by Hierholzer's algorithm.  This cannot happen on \(\mathsf{Good}\), as proved
above.  If \(\mathsf{Bad}\) occurs, the algorithm may either fall back
or run on an expansion whose pointwise guarantee is unavailable to the
analysis.  In either case we upper bound its contribution pessimistically:
replacing the conditional output law on \(\mathsf{Bad}\) by any other
distribution changes the final law by at most \(\P{\mathsf{Bad}}\) in total
variation distance.  Since \(\P{\mathsf{Bad}}\le \eps/16\), this event
contributes at most \(\eps/16\) to the final error.

On \(\mathsf{Good}\), apply \cref{thm:degree-two-sampler} to \(G^\star\) with
accuracy \(\eps/2\), and project the sampled expanded tour by contracting
gadget trajectories and undoing the degree-one suppressions.  This
deterministic projection cannot increase total variation distance.  Thus the
\(\eps/2\) error from sampling on \(G^\star\) remains at most \(\eps/2\) after
projection.
By \cref{thm:faithful-expansion}, the stationary distribution on \(G^\star\),
after projection, is within
\[
  e^{2\sum_v\eta_v}-1
  \le
  e^{\eps/8}-1
  \le
  \eps/4
\]
total variation distance of uniform on \(\ET(G)\), using \(0<\eps<1\).
Combining this projection error with the \(\eps/2\) sampling error and the
\(\eps/16\) construction-failure contribution gives total variation distance at
most \(\eps\), and total time
\[
  O\!\left(M^{3/2}\polylog(M/\eps)\right)
  =
  O\!\left(m^{3/2}\polylog(m/\eps)\right).
\]
The space usage is \(O(M)\): store the expanded graph, the projection data
needed to undo suppressions and contract gadget trajectories, and the
degree-two data structure from \cref{thm:degree-two-sampler}.  Since
\(M=O(m\polylog(m/\eps))\), this is \(O(m\polylog(m/\eps))\) space.
\end{proof}

\section{Stochastic Covering and Hurwitz-Stable Conjectures}\label{sec:hurwitz}

The sampler and its runtime analysis are complete after \cref{sec:gadgets}.  We
end with a structural consequence of the same subset-vector viewpoint.  This
result is independent of the algorithm, but it helps explain why the
skew-determinantal measures in the mixing theorem behave like a parity analogue
of negatively dependent measures.

The consequence is a local coupling statement in the spirit of stochastic
covering.  Stochastic covering was introduced to formalize the idea that two
neighboring one-coordinate conditionings of a negatively dependent measure can
be coupled locally \cite{PemantlePeres2014}.  Strongly Rayleigh measures,
introduced by Borcea, Branden, and Liggett via real-stable generating
polynomials \cite{BorceaBrandenLiggett2009}, satisfy this property: after
conditioning one coordinate to be \(1\) versus \(0\), the two conditional
measures can be coupled by changing only one other coordinate in the monotone
direction \cite{PemantlePeres2014,AnariOveisGharanRezaei2016}.

Here the support has a fixed parity, so the analogous local statement has a
parity shift.  We couple a sample with \(i=0\) to a sample with \(i=1\) so that
all but one of the remaining bits agree.  Equivalently, after deleting
coordinate \(i\), the problem is to transport an even subset of
\(W=[n]\setminus\{i\}\) to an odd subset of \(W\) by flipping one coordinate of
\(W\).  Re-inserting \(i\) then gives a coupling whose two original
configurations differ in exactly two coordinates.

The proof has two steps.  First, we represent each conditional law as the
squared coordinates of a unit vector in the subset-vector space with coordinate
\(i\) removed.  Second, we show that the two unit vectors are related by one
signed coordinate-flip operator \(c_b\).  The operator \(c_b\) is an isometry,
and every nonzero matrix entry of \(c_b\) connects two subsets that differ in one
remaining coordinate.  Squaring coordinates turns the two vectors into the two
probability distributions we want to couple.  A max-flow/min-cut argument then
turns this sparse isometry into an actual coupling supported on those same
one-coordinate moves after deleting \(i\), hence on two-coordinate moves in the
original cube.
Thus the argument has the same transport shape as stochastic covering for
strongly Rayleigh measures: an algebraic certificate gives local support for a
transport, and max-flow/min-cut turns this support statement into an actual
coupling.  The difference is the parity shift.  Strongly Rayleigh covering is
monotone, whereas here the two conditionings live on opposite parities after
deleting \(i\).  The local move therefore flips one remaining coordinate, as
shown in \cref{fig:parity-covering}.

\begin{figure}[t]
  \centering
  \begin{tikzpicture}[
    >=Latex,
    every node/.style={font=\small},
    state/.style={
      draw=black!55,
      rounded corners=2pt,
      fill=black!2,
      align=center,
      text width=3.2cm,
      minimum height=1.05cm,
      inner sep=4pt
    },
    reduced/.style={
      draw=teal!60!black,
      rounded corners=2pt,
      fill=teal!3,
      align=center,
      text width=3.2cm,
      minimum height=1.05cm,
      inner sep=4pt
    },
    down/.style={-{Latex[length=2mm,width=1.4mm]}, line width=.55pt, black!60},
    move/.style={-{Latex[length=2.4mm,width=1.6mm]}, line width=.8pt, teal!70!black},
    note/.style={font=\scriptsize, text=black!65}
  ]
    \node[note, text=blue!65!black] at (0,2.35) {condition on \(i=0\)};
    \node[note, text=teal!65!black] at (6.25,2.35) {condition on \(i=1\)};

    \node[state] (s0) at (0,1.25)
      {\(S_0=R\subseteq W\)\\\(i\) is absent};
    \node[state] (s1) at (6.25,1.25)
      {\(S_1=(R\symdiff\{j\})\cup\{i\}\)\\\(i\) is present};

    \node[reduced] (r0) at (0,-1.15)
      {\(R\subseteq W\)\\even parity};
    \node[reduced] (r1) at (6.25,-1.15)
      {\(R\symdiff\{j\}\subseteq W\)\\odd parity};

    \draw[down] (s0) -- node[midway, left=3pt, note] {delete \(i\)} (r0);
    \draw[down] (s1) -- node[midway, right=3pt, note] {delete \(i\)} (r1);
    \draw[move] (r0) -- node[midway, above=2pt, fill=white, inner sep=1.4pt,
      note, text=teal!70!black] {one flip in \(W\)} (r1);

    \node[draw=black!35, rounded corners=1.5pt, fill=white, align=center,
      inner sep=3pt, note] at (3.125,-2.15)
      {after reinserting \(i\),\\\(\card{S_0\symdiff S_1}=2\)};
  \end{tikzpicture}
  \caption{Parity stochastic covering.  After removing the conditioned
  coordinate \(i\), the two conditionings live on opposite parities of
  \(W=[n]\setminus\{i\}\).  The transport uses only edges that flip one
  coordinate in \(W\); re-inserting \(i\) on the right makes the full
  configurations differ in exactly two coordinates.}
  \label{fig:parity-covering}
\end{figure}
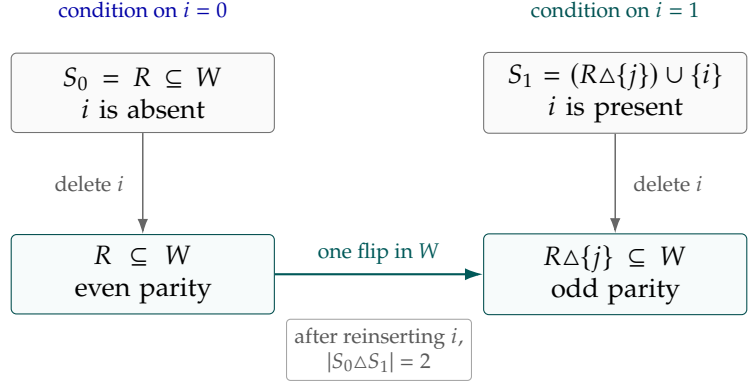

We first use a small transport lemma.  It is a linear-algebraic version of the
following elementary idea: if an isometry \(U\) has nonzero entries only on the
edges of a bipartite graph \(H\), then the probability distribution obtained by
squaring the coordinates of a vector \(x\) can be transported to the
distribution obtained by squaring the coordinates of \(Ux\), using only edges
of \(H\).  The isometry condition says that no set of left vertices can carry
more squared mass than its neighborhood carries on the right.  These are exactly
the weighted Hall inequalities, and max-flow/min-cut converts them into a
coupling.

\begin{lemma}[Orthogonal maps give local couplings]\label{lem:orthogonal-transport}
Let \(U:\R^L\to\R^R\) be a linear isometry between two Euclidean spaces with
fixed coordinate bases, so \(U^*U=I\), and let \(H\subseteq L\times R\) be a bipartite graph
such that \(U_{ba}=0\) whenever \((a,b)\notin H\).  For every unit vector
\(x\in\R^L\),
the two distributions
\[
  p_a=\abs{x_a}^2,
  \qquad
  q_b=\abs{(Ux)_b}^2
\]
admit a coupling supported on the edges of \(H\).
\end{lemma}

\begin{proof}
Since \(x\) is unit and \(U\) is an isometry, both \(p\) and \(q\) have total
mass one.
Put a source-to-\(a\) edge of capacity \(p_a\), a \(b\)-to-sink edge of
capacity \(q_b\), and infinite-capacity edges along \(H\).  By the
max-flow/min-cut theorem, equivalently the weighted Hall condition
\cite{FordFulkerson1956}, a unit flow, and hence the desired coupling, exists
if for every \(L_0\subseteq L\),
\[
  \sum_{a\in L_0}p_a
  \le
  \sum_{b\in N(L_0)}q_b,
\]
where \(N(L_0)\) is the neighborhood of \(L_0\) in \(H\).  Let \(P_{L_0}\) and
\(P_{N(L_0)}\) denote coordinate projections.  The support assumption gives
\[
  P_{N(L_0)^c}UP_{L_0}=0.
\]
Taking adjoints also gives \(P_{L_0}U^*P_{N(L_0)^c}=0\).  Therefore
\[
  P_{L_0}x
  =
  P_{L_0}U^*Ux
  =
  P_{L_0}U^*P_{N(L_0)}Ux,
\]
and hence
\[
  \sum_{a\in L_0}p_a
  =
  \norm{P_{L_0}x}^2
  \le
  \norm{P_{N(L_0)}Ux}^2
  =
  \sum_{b\in N(L_0)}q_b.
\]
Here the inequality uses that \(P_{L_0}\) is an orthogonal projection and
\(U^*\) has operator norm at most \(1\), since \(U\) is an isometry.
This proves the weighted Hall inequalities and hence the desired coupling.
\end{proof}

\begin{lemma}[Conditioning a Pfaffian vector]\label{lem:conditioned-vectors}
Let \(\Psi=\Psi_A\) be the Pfaffian vector of a real skew-symmetric matrix
\(A\), and fix a coordinate \(i\).  Let
\(W=[n]\setminus\{i\}\), let \(B=A_{W,W}\), and let \(\Psi_B\) be the
Pfaffian vector of \(B\) in the subset-vector space on \(W\).  After
normalization, the squared coordinates of \(\Psi_B\) give the law conditioned
on \(i\notin S\), with coordinate \(i\) deleted.  After normalization, the
squared coordinates of \(\iota_i\Psi_A\) give the law conditioned on
\(i\in S\), again with coordinate \(i\) deleted.  If both conditionings have
positive mass, then the normalized vectors
\[
  \Psi_0=\frac{\Psi_B}{\norm{\Psi_B}},
  \qquad
  \Psi_1=\frac{\iota_i\Psi_A}{\norm{\iota_i\Psi_A}}
\]
satisfy
\[
  \Psi_1=\pm c_b\Psi_0
\]
for some unit vector \(b\) supported on the coordinates other than \(i\).
\end{lemma}

\begin{proof}
The first two assertions are just the probabilistic meaning of conditioning in
the vector model.  Deleting coordinate \(i\) from a set not containing \(i\)
leaves the Pfaffian vector of the principal submatrix \(B\).  Deleting \(i\)
from a set containing \(i\) is represented by applying the deletion operator
\(\iota_i\) to \(\Psi_A\).
The squared norm of \(\Psi_B\) is the total weight of sets not containing
\(i\), and the squared norm of \(\iota_i\Psi_A\) is the total weight of sets
containing \(i\).  Thus the two normalizations in the statement are exactly the
normalizations of the two conditional laws, and the positive-mass assumption
means both vectors are nonzero.

It remains to relate the two conditional vectors.  This is the only Pfaffian
calculation in the proof, and it consists of two coefficient checks.  The signs
in these checks are fixed by the same insertion and deletion convention used in
\cref{lem:row-isotropy}; they are not additional choices.

The first check is ordinary expansion of a Pfaffian along the distinguished
coordinate \(i\).  To see the signs without extra notation, temporarily reorder
the coordinates so that \(i\) is first and the coordinates of \(W\) keep their
relative order.  Then the standard Pfaffian expansion gives, for every odd
\(R\subseteq W\),
\[
  \pf(A_{\{i\}\cup R,\{i\}\cup R})
  =
  \sum_{j\in R}
  (-1)^{\operatorname{pos}_R(j)+1}
  A_{ij}\,\pf(B_{R\setminus\{j\},R\setminus\{j\}}),
\]
where \(\operatorname{pos}_R(j)\) is the position of \(j\) in the ordered list
of \(R\).  Returning to the global coordinate order only multiplies each
coordinate by the deterministic signs coming from the signed deletion
\(\iota_i\) and signed insertion operators \(\varepsilon_j\).  Equivalently,
there is a vector \(a\in\R^W\), obtained from the \(i\)th row of \(A\) by these
fixed order signs, such that
\begin{equation}\label{eq:conditioned-vector-expansion}
  \iota_i\Psi_A=\varepsilon_a\Psi_B.
\end{equation}
In words, deleting \(i\) from the full Pfaffian vector is the same as inserting
one of the remaining coordinates into the smaller Pfaffian vector, with
coefficients given by the \(i\)th row of the matrix in the fixed subset ordering.

The second check is the analogous deletion identity inside the smaller matrix:
for every vector \(h\in\R^W\),
\begin{equation}\label{eq:deletion-pfaffian-identity}
  \iota_h\Psi_B=\varepsilon_{-Bh}\Psi_B.
\end{equation}
This identity says that, on a Pfaffian vector, a deletion can be rewritten as a
signed insertion with coefficients obtained from the matrix \(B\).
To verify it, compare the coefficient of \(e_R\) for an odd set
\(R\subseteq W\).  Write
\[
  \operatorname{ins}(j,U)=(-1)^{\card{\set{k\in U\given k<j}}},
\]
so that both \(\varepsilon_j e_U=\operatorname{ins}(j,U)e_{U\cup\{j\}}\) and
\(\iota_j e_{U\cup\{j\}}=\operatorname{ins}(j,U)e_U\).
The coefficient of \(e_R\) in \(\iota_h\Psi_B\) is therefore
\[
  \sum_{j\notin R}
  h_j\,\operatorname{ins}(j,R)\,
  \pf(B_{R\cup\{j\},R\cup\{j\}}).
\]
The coefficient of \(e_R\) in \(\varepsilon_{-Bh}\Psi_B\) is
\[
  -\sum_{k\in R}
  \operatorname{ins}(k,R\setminus\{k\})
  \left(\sum_{j\in W} B_{kj}h_j\right)
  \pf(B_{R\setminus\{k\},R\setminus\{k\}}).
\]
Fix \(j\notin R\).  The coefficient of \(h_j\) in the second display is the
Pfaffian expansion of \(\pf(B_{R\cup\{j\},R\cup\{j\}})\) along the row and
column \(j\), with exactly the insertion sign
\(\operatorname{ins}(j,R)\); hence it matches the first display.  Fix instead
\(j\in R\).  The coefficient of \(h_j\) in the second display is the Pfaffian
expansion of the odd principal Pfaffian \(\pf(B_{R,R})\) along \(j\), and this
Pfaffian is zero by convention.  Hence every coefficient agrees, proving
\cref{eq:deletion-pfaffian-identity}.

Since \(B\) is real skew-symmetric, its complex eigenvalues are purely
imaginary \cite{HornJohnson2012}.  Hence \(I-B\) is invertible.  Choose
\(h=(I-B)^{-1}a\).  By
\cref{eq:deletion-pfaffian-identity},
\[
  c_h\Psi_B
  =
  (\varepsilon_h+\iota_h)\Psi_B
  =
  \varepsilon_{(I-B)h}\Psi_B
  =
  \varepsilon_a\Psi_B
  =
  \iota_i\Psi_A,
\]
where the last equality is \cref{eq:conditioned-vector-expansion}.
Because the conditioning on \(i\in S\) has positive mass, the vector
\(\iota_i\Psi_A\) is nonzero, and hence \(h\ne0\).  Since
\(c_h^2=\norm{h}^2I\), we have
\[
  \norm{\iota_i\Psi_A}
  =
  \norm{c_h\Psi_B}
  =
  \norm{h}\,\norm{\Psi_B}.
\]
Thus with \(b=h/\norm{h}\), the normalized vectors satisfy
\(\Psi_1=c_b\Psi_0\), up to the global sign convention for Pfaffians.
\end{proof}

\begin{theorem}[Parity stochastic covering for skew-determinantal measures]\label{thm:pfaffian-covering}
Let \(A\) be real skew-symmetric and let \(\mu_A\) be the probability
distribution on subsets with weights
\(\mu_A(S)\propto\det(A_{S,S})\).  Thus \(\mu_A\) is supported on even
subsets.  Fix a coordinate \(i\) such that both conditionings have positive
mass.  Then there is a coupling \((S_0,S_1)\) of the conditional laws
\[
  \mu_A(\,\cdot\mid i\notin S)
  \quad\text{and}\quad
  \mu_A(\,\cdot\mid i\in S)
\]
such that, almost surely,
\[
  i\notin S_0,\qquad i\in S_1,
  \qquad
  \card{S_0\symdiff S_1}=2.
\]
\end{theorem}

\begin{proof}
Let \(W=[n]\setminus\{i\}\), let \(\cH_W\) be the subset-vector space on
coordinates in \(W\), set \(B=A_{W,W}\), and let \(\Psi_B\) be the Pfaffian
vector of \(B\) in \(\cH_W\).  Let
\[
  \Psi_0=\frac{\Psi_B}{\norm{\Psi_B}},
  \qquad
  \Psi_1=\frac{\iota_i\Psi_A}{\norm{\iota_i\Psi_A}}.
\]
The squared coordinates of \(\Psi_0\) are the conditional distribution given
\(i\notin S\), with coordinate \(i\) removed; these subsets of \(W\) have even
parity.  The squared coordinates of \(\Psi_1\) are the conditional distribution
given \(i\in S\), again with coordinate \(i\) removed; these subsets of \(W\)
have odd parity.  The positivity assumption in the theorem is exactly what
makes the two displayed normalizations nonzero.  Thus it remains to couple an even
subset of \(W\) to an odd subset of \(W\) while changing one coordinate.

By \cref{lem:conditioned-vectors}, \(\Psi_1=\pm c_b\Psi_0\) for some unit vector
\(b\in\R^W\).  The relations from \cref{sec:pfaffian-walk} give \(c_b^*=c_b\) and
\(c_b^2=\norm{b}^2I=I\), so \(c_b\) is an orthogonal map from
\(\cH_W^{\mathrm{even}}\) to \(\cH_W^{\mathrm{odd}}\).  Since
\[
  c_b=\sum_{j\ne i} b_jc_j,
\]
and each \(c_j\) flips only coordinate \(j\), its matrix in the coordinate
bases is supported only on pairs of subsets differing in exactly one element.
Apply \cref{lem:orthogonal-transport} with the left vertices equal to the even
subsets of \(W\), the right vertices equal to the odd subsets of \(W\), and
edges between pairs whose symmetric difference has size \(1\).  We take
\(U=\pm c_b\) and \(x=\Psi_0\).  The sign does not affect the squared
coordinates.  The resulting coupling on \(W\) changes exactly one coordinate
whenever it moves positive mass.  Re-inserting \(i\) on the \(\Psi_1\) side
gives a coupling of the original conditionings at Hamming distance two.
\end{proof}

This theorem is the parity analogue of stochastic covering.  It suggests that
the right ambient class may be Hurwitz-stable measures rather than strongly
Rayleigh measures.  In the conjectures that follow, disconnected supports are
handled in the usual way: restrict the walk to a connected component, or
equivalently ask for the same bounds on every component.  This caveat is
necessary for any spectral-gap statement, since a disconnected support gives
zero gap.

\begin{conjecture}[Hurwitz-stable parity covering]\label{conj:hurwitz-covering}
Let \(\mu\) be a probability distribution on \(\bits^n\) whose generating polynomial
\[
  g_\mu(z)=\sum_S\mu(S)z^S
\]
has nonnegative coefficients and is Hurwitz stable, i.e. \(g_\mu(z)\ne0\)
whenever \(\Re z_i>0\) for all \(i\).  If \(\mu\) is supported on one parity
class, then for every coordinate \(i\) for which both conditionings have
positive mass, the two conditionings \(X_i=0\) and \(X_i=1\) can be coupled so
that the full configurations differ in Hamming distance two.
\end{conjecture}

\begin{conjecture}[Hurwitz-stable flip--repair mixing]\label{conj:hurwitz-mixing}
Let \(\mu\) be a Hurwitz-stable distribution supported on one parity class.
On every connected component of its two-step parity flip--repair graph, the
parity flip--repair walk has spectral gap \(\Omega(1/n)\).  This is the walk
that flips a uniformly random coordinate to reach the opposite parity layer,
and then samples from \(\mu\) conditioned on being one flip away from that
intermediate point; the conditioning is always nonempty on the component
because it contains the state from which the first flip was made.  If \(\mu\)
is flat, meaning all positive atoms of \(\mu\) have equal mass, then the same
component walk
satisfies a modified log-Sobolev inequality, or an ordinary log-Sobolev
inequality, with constant \(n\polylog n\).
\end{conjecture}

\Cref{thm:spectral-comparison,thm:flat-lsi,thm:pfaffian-covering} prove
\cref{conj:hurwitz-covering,conj:hurwitz-mixing} for the skew-determinantal subclass,
after normalizing the coefficients to make a probability distribution.  For
this subclass,
\[
  g_A(z)=\sum_S\det(A_{S,S})z^S=\det(I+A\diag(z)),
\]
where \(A\) is real skew-symmetric.  The equality with the determinant is the
principal-minor expansion of \(\det(I+A\diag(z))\): choosing the nonidentity
rows and columns in a set \(S\) contributes \(\det(A_{S,S})\prod_{i\in S}z_i\).
This also makes the nonnegativity of the coefficients transparent, since
\(\det(A_{S,S})=\pf(A_{S,S})^2\ge0\).

We now check Hurwitz stability.  If \(\Re z_i>0\), then each \(z_i\) is nonzero
and
\[
  \det(I+\diag(z)A)=\det(\diag(z))\det(\diag(z)^{-1}+A).
\]
Moreover, the identity \(\det(I+XY)=\det(I+YX)\) gives
\(\det(I+A\diag(z))=\det(I+\diag(z)A)\) \cite{HornJohnson2012}.  The matrix
\(\diag(z)^{-1}+A\) has positive definite Hermitian part: the Hermitian part is
\(\diag(\Re(1/z_i))\), since \(A^*=-A\), and \(\Re(1/z_i)>0\) whenever
\(\Re z_i>0\).  If
\((\diag(z)^{-1}+A)x=0\) for some nonzero complex vector \(x\), then
\[
  x^*(\diag(z)^{-1}+A)x=0.
\]
But the real part of the left-hand side is
\(\sum_i \Re(1/z_i)\abs{x_i}^2>0\), a contradiction.
Thus \(\diag(z)^{-1}+A\) is nonsingular, and hence \(g_A(z)\ne0\).

\section*{Acknowledgments}

We thank Thuy-Duong Vuong, Tianyu Liu, and Eric Ma for discussions about the problem.

\appendix

\section{Provenance}\label{app:provenance}

This manuscript used AI assistance in a substantive way.  The author had
already devised the high-level plan for the sampler, including the reduction to
degree two, the transition-system walk, the switching-network gadget reduction,
and the dynamic data-structure implementation.  The author conjectured the
stated mixing bounds for the skew-determinantal flip--repair walk, not the proof
route used to prove them.

The proof of the mixing theorem, including the square-root projection
formulation, the excitation-number comparison, and the flat log-Sobolev proof
via entropic uncertainty and Rothaus centering, was produced by GPT 5.5 Pro
Extended.  Codex then assembled the paper, integrated the proof into the
manuscript, added surrounding exposition and references, and typeset the draft
under the author's supervision.

\PrintBibliography


@article{deBruijnAardenneEhrenfest1951,
  author = {de Bruijn, N. G. and van Aardenne-Ehrenfest, T.},
  title = {Circuits and Trees in Oriented Linear Graphs},
  journal = {Simon Stevin},
  volume = {28},
  pages = {203--217},
  year = {1951}
}

@article{Hierholzer1873,
  author = {Hierholzer, Carl},
  title = {Ueber die M{\"o}glichkeit, einen Linienzug ohne Wiederholung und ohne Unterbrechung zu umfahren},
  journal = {Mathematische Annalen},
  volume = {6},
  pages = {30--32},
  year = {1873},
  doi = {10.1007/BF01442866}
}

@book{Tutte1984,
  author = {Tutte, W. T.},
  title = {Graph Theory},
  series = {Encyclopedia of Mathematics and its Applications},
  volume = {21},
  publisher = {Addison-Wesley},
  year = {1984}
}

@article{SmithTutte1941,
  author = {Smith, C. A. B. and Tutte, W. T.},
  title = {On Unicursal Paths in a Network of Degree 4},
  journal = {The American Mathematical Monthly},
  volume = {48},
  number = {4},
  pages = {233--237},
  year = {1941},
  doi = {10.1080/00029890.1941.11991103}
}

@article{CohnLempel1972,
  author = {Cohn, M. and Lempel, A.},
  title = {Cycle Decomposition by Disjoint Transpositions},
  journal = {Journal of Combinatorial Theory, Series A},
  volume = {13},
  number = {1},
  pages = {83--89},
  year = {1972},
  doi = {10.1016/0097-3165(72)90010-6}
}

@article{Traldi2011,
  author = {Traldi, Lorenzo},
  title = {Binary Nullity, Euler Circuits and Interlace Polynomials},
  journal = {European Journal of Combinatorics},
  volume = {32},
  number = {6},
  pages = {944--950},
  year = {2011},
  doi = {10.1016/j.ejc.2011.02.004},
  eprint = {0903.4405},
  archivePrefix = {arXiv},
  primaryClass = {math.CO}
}

@article{BouchetCunninghamGeelen1998,
  author = {Bouchet, Andr{\'e} and Cunningham, William H. and Geelen, Jim},
  title = {Principally Unimodular Skew-Symmetric Matrices},
  journal = {Combinatorica},
  volume = {18},
  number = {4},
  pages = {461--486},
  year = {1998},
  doi = {10.1007/s004930050033}
}

@article{Bouchet1987Unimodularity,
  author = {Bouchet, Andr{\'e}},
  title = {Unimodularity and Circle Graphs},
  journal = {Discrete Mathematics},
  volume = {66},
  number = {1--2},
  pages = {203--208},
  year = {1987},
  doi = {10.1016/0012-365X(87)90132-4}
}

@article{Bouchet1992UnimodularOrientations,
  author = {Bouchet, Andr{\'e}},
  title = {A Characterization of Unimodular Orientations of Simple Graphs},
  journal = {Journal of Combinatorial Theory, Series B},
  volume = {56},
  number = {1},
  pages = {45--54},
  year = {1992},
  doi = {10.1016/0095-8956(92)90005-I}
}

@article{KandelMatiasUngerWinkler1996,
  author = {Kandel, Denise and Matias, Yossi and Unger, Ron and Winkler, Peter},
  title = {Shuffling Biological Sequences},
  journal = {Discrete Applied Mathematics},
  volume = {71},
  number = {1--3},
  pages = {171--185},
  year = {1996},
  doi = {10.1016/S0166-218X(97)81456-4}
}

@article{ProppWilson1998,
  author = {Propp, James Gary and Wilson, David Bruce},
  title = {How to Get a Perfectly Random Sample from a Generic {Markov} Chain and Generate a Random Spanning Tree of a Directed Graph},
  journal = {Journal of Algorithms},
  volume = {27},
  number = {2},
  pages = {170--217},
  year = {1998},
  doi = {10.1006/jagm.1997.0917}
}

@article{JiangEtAl2008,
  author = {Jiang, Minghui and Anderson, James and Gillespie, Joel and Mayne, Martin},
  title = {{uShuffle}: A Useful Tool for Shuffling Biological Sequences while Preserving the \(k\)-let Counts},
  journal = {BMC Bioinformatics},
  volume = {9},
  number = {192},
  year = {2008},
  doi = {10.1186/1471-2105-9-192}
}

@article{BoczkowskiPeresSousi2018,
  author = {Boczkowski, Lucas and Peres, Yuval and Sousi, Perla},
  title = {Sensitivity of Mixing Times in Eulerian Digraphs},
  journal = {SIAM Journal on Discrete Mathematics},
  volume = {32},
  number = {1},
  pages = {624--655},
  year = {2018},
  doi = {10.1137/16M1073376}
}

@inproceedings{TetaliVempala1997,
  author = {Tetali, Prasad and Vempala, Santosh},
  title = {Random Sampling of Euler Tours},
  booktitle = {Randomization and Approximation Techniques in Computer Science},
  series = {Lecture Notes in Computer Science},
  volume = {1269},
  pages = {57--66},
  publisher = {Springer},
  year = {1997},
  doi = {10.1007/3-540-63248-4_6}
}

@article{CreedCryan2013,
  author = {Creed, P{\'a}id{\'i} and Cryan, Mary},
  title = {The Number of Euler Tours of Random Directed Graphs},
  journal = {Electronic Journal of Combinatorics},
  volume = {20},
  number = {3},
  pages = {P13},
  year = {2013},
  doi = {10.37236/2377}
}

@article{Morris2009Thorp,
  author = {Morris, Ben},
  title = {The Mixing Time of the {Thorp} Shuffle},
  journal = {SIAM Journal on Computing},
  volume = {38},
  number = {2},
  pages = {484--504},
  year = {2008},
  doi = {10.1137/050636231}
}

@article{Benes1964,
  author = {Bene{\v{s}}, V. E.},
  title = {Optimal Rearrangeable Multistage Connecting Networks},
  journal = {The Bell System Technical Journal},
  volume = {43},
  number = {4},
  pages = {1641--1656},
  year = {1964},
  doi = {10.1002/j.1538-7305.1964.tb04103.x}
}

@article{Waksman1968,
  author = {Waksman, Abraham},
  title = {A Permutation Network},
  journal = {Journal of the ACM},
  volume = {15},
  number = {1},
  pages = {159--163},
  year = {1968},
  doi = {10.1145/321439.321449}
}

@article{DiaconisShahshahani1981,
  author = {Diaconis, Persi and Shahshahani, Mehrdad},
  title = {Generating a Random Permutation with Random Transpositions},
  journal = {Zeitschrift f{\"u}r Wahrscheinlichkeitstheorie und Verwandte Gebiete},
  volume = {57},
  number = {2},
  pages = {159--179},
  year = {1981},
  doi = {10.1007/BF00535487}
}

@article{Morris2009Improved,
  author = {Morris, Ben},
  title = {Improved Mixing Time Bounds for the {Thorp} Shuffle and {L}-Reversal Chain},
  journal = {Annals of Probability},
  volume = {37},
  number = {2},
  pages = {453--477},
  year = {2009},
  doi = {10.1214/08-AOP409},
  eprint = {0912.2759},
  archivePrefix = {arXiv},
  primaryClass = {math.PR}
}

@inproceedings{Czumaj2015SwitchingNetworks,
  author = {Czumaj, Artur},
  title = {Random Permutations Using Switching Networks},
  booktitle = {Proceedings of the Forty-Seventh Annual ACM Symposium on Theory of Computing},
  pages = {703--712},
  publisher = {Association for Computing Machinery},
  year = {2015},
  doi = {10.1145/2746539.2746629}
}

@article{GeStefankovic2012,
  author = {Ge, Qi and {\v{S}}tefankovi{\v{c}}, Daniel},
  title = {The Complexity of Counting {Eulerian} Tours in 4-Regular Graphs},
  journal = {Algorithmica},
  volume = {63},
  number = {3},
  pages = {588--601},
  year = {2012},
  doi = {10.1007/s00453-010-9463-4},
  eprint = {1009.5019},
  archivePrefix = {arXiv},
  primaryClass = {cs.CC}
}

@inproceedings{AnariHuangLiuVuongXuYu2023,
  author = {Anari, Nima and Huang, Yizhi and Liu, Tianyu and Vuong, Thuy-Duong and Xu, Brian and Yu, Katherine},
  title = {Parallel Discrete Sampling via Continuous Walks},
  booktitle = {Proceedings of the 55th Annual ACM Symposium on Theory of Computing},
  pages = {103--116},
  publisher = {Association for Computing Machinery},
  year = {2023},
  doi = {10.1145/3564246.3585207}
}

@inproceedings{AnariLiuOveisGharanVinzantVuong2021,
  author = {Anari, Nima and Liu, Kuikui and Oveis Gharan, Shayan and Vinzant, Cynthia and Vuong, Thuy-Duong},
  title = {Log-Concave Polynomials {IV}: Approximate Exchange, Tight Mixing Times, and Near-Optimal Sampling of Forests},
  booktitle = {Proceedings of the 53rd Annual ACM SIGACT Symposium on Theory of Computing},
  pages = {408--420},
  publisher = {Association for Computing Machinery},
  year = {2021},
  doi = {10.1145/3406325.3451091},
  eprint = {2004.07220},
  archivePrefix = {arXiv},
  primaryClass = {cs.DS}
}

@inproceedings{ChenKyngLiuPengProbstGutenbergSachdeva2022,
  author = {Chen, Li and Kyng, Rasmus and Liu, Yang P. and Peng, Richard and Probst Gutenberg, Maximilian and Sachdeva, Sushant},
  title = {Maximum Flow and Minimum-Cost Flow in Almost-Linear Time},
  booktitle = {Proceedings of the 63rd Annual IEEE Symposium on Foundations of Computer Science},
  pages = {612--623},
  publisher = {IEEE},
  year = {2022},
  doi = {10.1109/FOCS54457.2022.00064},
  eprint = {2203.00671},
  archivePrefix = {arXiv},
  primaryClass = {cs.DS}
}

@article{MaassenUffink1988,
  author = {Maassen, Hans and Uffink, J. B. M.},
  title = {Generalized Entropic Uncertainty Relations},
  journal = {Physical Review Letters},
  volume = {60},
  number = {12},
  pages = {1103--1106},
  year = {1988},
  doi = {10.1103/PhysRevLett.60.1103}
}

@article{Rothaus1981,
  author = {Rothaus, O. S.},
  title = {Logarithmic {Sobolev} Inequalities and the Spectrum of {Schr{\"o}dinger} Operators},
  journal = {Journal of Functional Analysis},
  volume = {42},
  number = {1},
  pages = {110--120},
  year = {1981},
  doi = {10.1016/0022-1236(81)90050-1}
}

@article{DiaconisSaloffCoste1996,
  author = {Diaconis, Persi and Saloff-Coste, Laurent},
  title = {Logarithmic {Sobolev} Inequalities for Finite {Markov} Chains},
  journal = {Annals of Applied Probability},
  volume = {6},
  number = {3},
  pages = {695--750},
  year = {1996},
  doi = {10.1214/aoap/1034968224}
}

@book{LevinPeresWilmer2017,
  author = {Levin, David A. and Peres, Yuval and Wilmer, Elizabeth L.},
  title = {Markov Chains and Mixing Times},
  edition = {2},
  publisher = {American Mathematical Society},
  year = {2017},
  isbn = {978-1-4704-2962-1}
}

@book{HornJohnson2012,
  author = {Horn, Roger A. and Johnson, Charles R.},
  title = {Matrix Analysis},
  edition = {2},
  publisher = {Cambridge University Press},
  year = {2012},
  doi = {10.1017/CBO9781139020411}
}

@book{Berezin1966,
  author = {Berezin, F. A.},
  title = {The Method of Second Quantization},
  publisher = {Academic Press},
  year = {1966}
}

@article{Bravyi2005,
  author = {Bravyi, Sergey},
  title = {Lagrangian Representation for Fermionic Linear Optics},
  journal = {Quantum Information and Computation},
  volume = {5},
  number = {3},
  pages = {216--238},
  year = {2005},
  eprint = {quant-ph/0404180},
  archivePrefix = {arXiv}
}

@article{Knuth1996,
  author = {Knuth, Donald E.},
  title = {Overlapping {Pfaffians}},
  journal = {The Electronic Journal of Combinatorics},
  volume = {3},
  number = {2},
  pages = {R5},
  year = {1996},
  doi = {10.37236/1263},
  eprint = {math/9503234},
  archivePrefix = {arXiv}
}

@article{FordFulkerson1956,
  author = {Ford, L. R. and Fulkerson, D. R.},
  title = {Maximal Flow Through a Network},
  journal = {Canadian Journal of Mathematics},
  volume = {8},
  pages = {399--404},
  year = {1956},
  doi = {10.4153/CJM-1956-045-5}
}

@article{SeidelAragon1996,
  author = {Seidel, Raimund and Aragon, Cecilia R.},
  title = {Randomized Search Trees},
  journal = {Algorithmica},
  volume = {16},
  number = {4--5},
  pages = {464--497},
  year = {1996},
  doi = {10.1007/BF01940876}
}

@article{BorceaBrandenLiggett2009,
  author = {Borcea, Julius and Br{\"a}nd{\'e}n, Petter and Liggett, Thomas M.},
  title = {Negative Dependence and the Geometry of Polynomials},
  journal = {Journal of the American Mathematical Society},
  volume = {22},
  number = {2},
  pages = {521--567},
  year = {2009},
  doi = {10.1090/S0894-0347-08-00618-8}
}

@article{PemantlePeres2014,
  author = {Pemantle, Robin and Peres, Yuval},
  title = {Concentration of {Lipschitz} Functionals of Determinantal and Other Strong {Rayleigh} Measures},
  journal = {Combinatorics, Probability and Computing},
  volume = {23},
  number = {1},
  pages = {140--160},
  year = {2014},
  doi = {10.1017/S0963548313000345},
  eprint = {1108.0687},
  archivePrefix = {arXiv},
  primaryClass = {math.PR}
}

@inproceedings{AnariOveisGharanRezaei2016,
  author = {Anari, Nima and Oveis Gharan, Shayan and Rezaei, Alireza},
  title = {Monte Carlo {Markov} Chain Algorithms for Sampling Strongly {Rayleigh} Distributions and Determinantal Point Processes},
  booktitle = {Proceedings of the 29th Annual Conference on Learning Theory},
  series = {Proceedings of Machine Learning Research},
  volume = {49},
  pages = {103--115},
  publisher = {PMLR},
  year = {2016},
  eprint = {1602.05242},
  archivePrefix = {arXiv},
  primaryClass = {cs.DS},
  url = {https://proceedings.mlr.press/v49/anari16.html}
}

@inproceedings{AnariAlimohammadiShiragurVuong2021,
  author = {Alimohammadi, Yeganeh and Anari, Nima and Shiragur, Kirankumar and Vuong, Thuy-Duong},
  title = {Fractionally Log-Concave and Sector-Stable Polynomials: Counting Planar Matchings and More},
  booktitle = {Proceedings of the 53rd Annual ACM SIGACT Symposium on Theory of Computing},
  pages = {433--446},
  publisher = {Association for Computing Machinery},
  year = {2021},
  doi = {10.1145/3406325.3451123},
  eprint = {2102.02708},
  archivePrefix = {arXiv},
  primaryClass = {cs.DS}
}

@inproceedings{ChenLiuVigoda2021,
  author = {Chen, Zongchen and Liu, Kuikui and Vigoda, Eric},
  title = {Spectral Independence via Stability and Applications to Holant-Type Problems},
  booktitle = {Proceedings of the 62nd Annual IEEE Symposium on Foundations of Computer Science},
  pages = {149--160},
  year = {2021},
  doi = {10.1109/FOCS52979.2021.00023},
  eprint = {2106.03366},
  archivePrefix = {arXiv},
  primaryClass = {cs.DS}
}

@article{ChenFengYinZhang2024,
  author = {Chen, Xiaoyu and Feng, Weiming and Yin, Yitong and Zhang, Xinyuan},
  title = {Rapid Mixing of {Glauber} Dynamics via Spectral Independence for All Degrees},
  journal = {SIAM Journal on Computing},
  volume = {53},
  number = {4},
  pages = {FOCS21-224--FOCS21-298},
  year = {2024},
  doi = {10.1137/22M1474734}
}
\end{document}